\documentclass[a4paper, 12pt]{article}

\usepackage[english]{babel}
\usepackage[sort&compress]{natbib}
\bibpunct{(}{)}{;}{a}{}{,} 

\usepackage{amsthm, amsmath, amssymb, mathrsfs, multirow, url, subfigure}
\usepackage{graphicx} 
\usepackage{ifthen} 
\usepackage{amsfonts}
\usepackage[usenames]{color}
\usepackage{fullpage}
\usepackage{booktabs}

\theoremstyle{plain}

\newtheorem{prop}{Proposition}

\theoremstyle{definition}

\newtheorem*{asmpS}{Assumption~S}

\theoremstyle{remark} 
\newtheorem{ex}{Example}

\newcommand{\prob}{\mathsf{P}} 
\newcommand{\E}{\mathsf{E}}

\newcommand{\nm}{{\sf N}}

\newcommand{\RR}{\mathbb{R}}

\newcommand{\ZZ}{\mathbb{Z}}

\newcommand{\eps}{\varepsilon}

\usepackage{stackrel}
\newcommand{\Tra}{^{\sf T}}
\newcommand{\Inv}{^{-1}}

\title{Generalized Bayes inference on a linear personalized minimum clinically important difference}
\author{Pei-Shien Wu\footnote{Department of Statistics, North Carolina State University; {\tt pwu9@ncsu.edu}, {\tt rgmarti3@ncsu.edu}} \quad and \quad Ryan Martin$^*$}
\date{\today}

\begin{document}

\maketitle 

\begin{abstract}

Inference on the minimum clinically important difference, or MCID, is an important practical problem in medicine.  The basic idea is that a treatment being statistically significant may not lead to an improvement in the patients' well-being.  The MCID is defined as a threshold such that, if a diagnostic measure exceeds this threshold, then the patients are more likely to notice an improvement.  Typical formulations use an underspecified model, which makes a genuine Bayesian solution out of reach.  Here, for a challenging personalized MCID problem, where the practically-significant threshold depends on patients' profiles, we develop a novel generalized posterior distribution, based on a working binary quantile regression model, that can be used for estimation and inference.  The advantage of this formulation is two-fold: we can theoretically control the bias of the misspecified model and it has a latent variable representation which we can leverage for efficient Gibbs sampling.  To ensure that the generalized Bayes inferences achieve a level of frequentist reliability, we propose a variation on the so-called generalized posterior calibration algorithm to suitably tune the spread of our proposed posterior.

\smallskip

\end{abstract}

\section{Introduction}
\label{S:intro}  

In clinical trials, {\em statistical significance} is widely used to infer treatment effects.  It is well-known, however, that statistically significant results may not accurately reflect the patients' experiences from receiving treatment. For example, even if it is determined that the treatment significantly reduces, e.g., blood pressure, the patient might feel little or no improvement after receiving the treatment, hence the treatment would not be considered is not {\em practically} or {\em clinically significant}. The so-called {\em minimum clinically important difference} (MCID), first proposed in \cite{jaeschke1989measurement}, aimed to incorporate patients' experiences into the evaluation of a treatment.  The basic idea or objective is to identify a cutoff such that, if the diagnostic measure (the one commonly used to determine statistical significance) exceeds this cutoff, then patients are expected to experience a change associated with the treatment.  This is accomplished by incorporating a ``patient-reported'' treatment outcome into the analysis along with the diagnostic measure.  With this additional information and change in objective, the MCID cutoff could end up being considerably different than that used to determine statistical significance.  In any case, the MCID is a relevant quantity and, as such, methods to estimate and quantify uncertainty about the MCID are needed.  

Early efforts were made, largely in the medical statistics literature, to estimate the MCID, but the methods were somewhat ad hoc \citep{parker2012determination,copay2008minimum}.  A systematic investigation appeared only relatively recently, in \cite{hedayat2015minimum}, where the MCID was defined formally in terms of a classification problem, in particular, as a minimizer of a suitably defined misclassification error.  This formulation, described in more detail in Section~\ref{sec:MCID_setup}, allows for rather complex relationships between the patient-reported outcomes, the diagnostic measures, and the MCID.  The simplest case, what \cite{hedayat2015minimum} call the {\em population MCID} problem, is where the MCID is simply a scalar cutoff, the same for all patients, a ``one size fits all'' level for clinical significance.  The more interesting and challenging case is what they called the {\em personalized MCID}, where the MCID cutoff can be expressed as a function of some patient-specific covariates.  In this way, the clinically significant cutoff can vary from patient to patient.  

The framework put forward in \cite{hedayat2015minimum}, and subsequently taken up by Jiwei Zhao and his group \citep[e.g.][]{zhou2020estimation, zhou2020interval, zhou2021statistical}, put the MCID problem on a firm statistical foundation, through the theory of M-estimation \citep[e.g.,][]{stefanski2002calculus}. But the fact that the loss function linking the patient-reported outcomes and corresponding diagnostic measures to the MCID threshold is not-smooth creates some challenges.  To avoid this, and to make the corresponding objective function more amenable to numerical and theoretical analyses, a natural idea is to suitably smooth the inherently discontinuous MCID loss.  While this smoothing helps in some ways, it hurts in others; in particular, smoothing changes the original optimization problem, which generally creates bias.  An advantage of a generalized Bayes framework over empirical risk minimization is that sampling from the (possibly relatively rough) generalized posterior distribution is often easier than optimization.  Indeed, for the population MCID problem, \cite{syring.martin.mcid} developed a so-called Gibbs posterior distribution for estimation and valid inference, no smoothing required.  For the personalized MCID problem, the added dimension creates some computational challenges for both optimization and sampling.   \cite{hedayat2015minimum} and \cite{zhou2020interval} describe smoothing strategies for risk minimization in the personalized MCID problem.
However, smoothing strategies designed to simplify optimization do not directly assist with sampling from the corresponding Gibbs posterior distribution.  That is, simply replacing a not-smooth loss with a smooth version does not imply existence of an efficient posterior sampling algorithm.  This begs the question: {\em what kind of smoothing will induce an efficient sampling algorithm?} To our knowledge, this has not been considered in the literature, at least not in the MCID literature.  This paper aims to fill this gap.  

A go-to computational strategy is {\em data-augmentation}, the introduction of latent variables in a principled way so that computations in the higher-dimensional problem are simpler and more efficient than a more complicated algorithm working on the lower-dimensional problem  \citep[e.g.][]{tanner1987calculation}.  A classical example of this is in binary regression, where Gaussian auxiliary variables are introduced for the purpose of fitting a probit model  \citep[e.g.][]{albert1993bayesian}.  More recently, \cite{polson2011data} developed a similar data-augmentation strategy for Bayesian-like inference in support vector machines, another type of binary regression/classification problem.  This work inspired us to consider a similar strategy for the personalized MCID problem. Following some background on the MCID problem in Section~\ref{sec:MCID_setup}, our jumping off point is the seminal work of \cite{manski1985semiparametric} and others on {\em binary quantile regression} or {\em BQR} \citep[e.g.,][]{benoit2012binary, mollica2017bayesian}.  In particular, in Section~\ref{S:bbqr} we show that, first, replacing the original MCID loss/empirical risk achieves a degree of smoothing superior to those in the aforementioned references, while still allowing for control on the bias introduced.  Second, this binary quantile regression formulation admits a latent variable representation that naturally leads to an efficient Gibbs sampling algorithm for drawing from our proposed generalized Bayes posterior for the MCID.  


Based on these developments, we have a clear path forward to defining a generalized posterior distribution for the personalized MCID.  Roughly, we simply use the likelihood function of this purposely-misspecified binary quantile regression model as a working likelihood in a typical Bayesian formulation.   It is important to recognize, however, that this model is surely misspecified, so, as \citet{kleijn2012bernstein} make clear, we cannot count on the asymptotic calibration enjoyed by Bayesian posterior distributions under correctly-specified models to ensure that, e.g., our generalized posterior credible intervals are valid.  In other words, it is up to us to ensure that validity is achieved and, for this, we use (a variation on) the {\em generalized posterior calibration} (GPC) strategy originally developed in \cite{syring2019calibrating}.  Section~\ref{S:general.post} provides the details of our proposed generalized posterior construction, calibration, and computation.  Then, in Section~\ref{SS:personalizedMCID_results}, we present some numerical results to compare the performance of our proposed generalized posterior distribution based on the working binary quantile regression model and GPC and the state of the art method for estimation and inference on the linear personalized MCID developed by \cite{zhou2021statistical} and implemented in their R package {\tt MCID} \citep{R.mcid}.  We show that our proposed ``BQR + GPC'' strategy has superior performance in terms of both estimation accuracy and credible interval coverage probability and average length across a range of examples.  Some concluding remarks are given in Section~\ref{S:discuss} and some additional technical details are presented in the Appendix.  

\section{Background}
\label{sec:MCID_setup}

\subsection{Problem formulation}

Consider data $(X,Y) \sim \prob$, where $Y\in\{-1,+1\}$ denotes is the patient-reported outcome, with $Y=1$ indicating the patient felt the treatment was effective, and $X \in \RR$ denotes some diagnostic measure, e.g., a change in blood pressure between before and after treatment.  As mentioned above, recall that the MCID is intuitively defined a threshold on the diagnostic measure beyond which the patients are likely to feel that the treatment was beneficial.  Based on this intuition, \citet{hedayat2015minimum} formally defined the (population) MCID as the minimizer of the risk or expected (weighted) loss function
\[ R(\theta) = \E\{ w(Y) \, \ell_\theta(X,Y) \}, \]
where 
\[ \ell_\theta(x,y) = \tfrac12 \{1 - y \, \text{sign}(x - \theta)\}, \quad (x,y) \in \RR \times \{-1,+1\}, \quad \theta \in \RR, \]
and the expectation is with respect to the joint distribution, $\prob$, of $(X,Y)$.  

\citet{hedayat2015minimum} take $w(y) \equiv 1$, whereas \citet{zhou2020estimation, zhou2021statistical} work with $w$ that satisfies 
\[ w(-1) = (1-\varpi)^{-1} \quad \text{and} \quad w(+1) = \varpi^{-1}, \]
where $\varpi = \prob(Y=+1)$ is the marginal probability of a positive patient reported outcome under $\prob$.  There are several equivalent ways to characterize the MCID, $\theta^\star = \arg\min R(\theta)$, the solution to the above risk minimization problem.  The version that we find most intuitive relies on the function 
\begin{equation}
\label{eq:m}
m(x) = \prob(Y = 1 \mid X=x), \quad x \in \RR, 
\end{equation}
the conditional probability of a positive patient reported outcome, given $X=x$, under $\prob$.  For the two formulations, Hedayat et al.~and Zhou et al., the MCID is shown to be $\theta^\star = \inf\{x: m(x) \geq \frac12\}$ and $\theta^\star = \inf\{x: m(x) \geq \varpi\}$, respectively.  Clearly, if $\varpi = \frac12$, i.e., if the design is {\em balanced} in the sense that positive and negative patient reported outcomes occur in equal proportions, then both formulations are targeting the same MCID.  In general, these two definitions correspond to two distinct interpretations:
\begin{itemize}
\item under Hedayat et al.'s formulation, $m(x) \geq \frac12$ means that a patient with $X=x$ is more likely than not to report the treatment is beneficial; 
\vspace{-2mm}
\item under Zhou et al.'s formulation, $m(x) \geq \varpi$ means that a patient with $X=x$ is more likely to report the treatment is beneficial than a typical patient. 
\end{itemize}
Our view is that both of these interpretations, and corresponding versions of the MCID, are meaningful and potentially relevant in a given application.\footnote{\citet{zhou2020estimation} argue in general favor their MCID by claiming that Hedayat et al.'s version is ``incompatible with imbalance,'' which we disagree with. It is true that Hedayat et al.'s MCID is not well-defined when $m$ is {\em uniformly} greater than or less than $\frac12$, but this is not implied by $\varpi \neq \frac12$. Moreover, FDA guidance on patient-reported outcomes studies (\url{https://www.fda.gov/media/77832/download}) states that, roughly, the design should be such that $m$ spans most/all of the interval $[0,1]$.} In this paper, however, we focus on the latter version, based on Zhou et al.'s formulation.  Of course, in the balanced case where $\varpi = \frac12$, the two approaches and corresponding MCIDs agree.  


Since $\prob$ is unknown, so too is the true MCID, $\theta^\star = \arg \min_\theta R(\theta)$.  However, if iid data $(X_i,Y_i)$, for $i=1,\ldots,n$, is available from $\prob$, then this can be used to estimate the MCID in a very natural way.  That is, define the estimator $\hat\theta_n \in \arg\min_\theta R_n(\theta)$, where 
\[ R_n(\theta) = \frac1n \sum_{i=1}^n w(Y_i) \, \ell_\theta(X_i, Y_i), \]
is the empirical version of the risk function.  This $\hat\theta_n$ defined above is an M-estimator, or an empirical risk minimizer, and its statistical properties have been studied in \cite{hedayat2015minimum}, \cite{syring.martin.mcid}, and  \cite{zhou2020estimation}.  

Beyond the population version of the MCID problem, \cite{hedayat2015minimum} proposed a practically relevant and technically interesting extension.  Suppose, in addition to diagnostic measure $X$ and the patient-reported outcomes $Y$, we also have access to a covariate $Z \in \RR^q$ that provides additional information about the patient's profile.  For example, in a knee surgery to repair torn cartilage, possible covariates would include patient's gender, age, treatment assignment, stage of knee damage, etc. Then this patient-level information encoded in $Z$ can be used to tailor a MCID threshold to the individual patients, a so-called {\em personalized MCID}.  That is, consider a new loss function 
\[ \ell_\theta(x,y,z) = \tfrac12 [1 - y \, \text{sign}\{x - \theta(z)\} ], \]
where now $\theta$ corresponds a function of the covariate.  Then the personalized MCID, $\theta^\star$, is still the minimizer of the expected loss, but now it is a function, and inference on the MCID cutoff for a patient having covariate level fixed at, say, $\tilde z$, can be made.  Some theory for the empirical risk minimizer in the context of a nonparametric personalized MCID can be found in \cite{hedayat2015minimum} and \cite{zhou2020interval}.  

Of course, estimating a genuine ``MCID function'' in practical applications might be too ambitious, but there are other ways this patient covariate information can be incorporated.  Indeed, a very natural approach is to consider a linear function, where the MCID function is determined by a finite-dimensional parameter vector, i.e., 
\[ \theta(z) = \beta^\top z. \]
This formulation has been adopted in \cite{zhou2020estimation, zhou2020interval}, and is the approach we will consider here.  This shifts the initial focus off the MCID function, to the parameter $\beta$ that determines it; any inferences that can be drawn about $\beta$ can readily be translated back to the MCID, the quantity of interest.  To be clear, there is a subtle shift in the notation, and now we express the loss function as 
\[ \ell_\beta(x,y,z) = \tfrac12 \{1 - y \, \text{sign}(x - \beta^\top z) \}, \quad \beta \in \RR^q. \]
Then \cite{zhou2020estimation, zhou2020interval} consider the risk function $R(\beta) = \E\{w(Y) \, \ell_\beta(X,Y,Z)\}$ as above, and $\beta^\star$ the corresponding risk minimizer, which determines the personalized MCID, or MCID function, $\theta^\star(z) = \beta^{\star \top} z$.  The estimator $\hat\beta_n$ can be defined accordingly as a minimizer of the empirical risk $R_n(\beta) = n^{-1} \sum_{i=1}^n w(Y_i) \, \ell_\beta(X_i, Y_i, Z_i)$, based on iid data $D_i = (X_i, Y_i, Z_i)$, for $i=1,\ldots,n$ from $\prob$.



As an alternative to the M-estimation framework, one can similarly construct a Gibbs posterior distribution for (probabilistic) inference about $\beta$.  This Gibbs posterior distribution is defined by the density function 
\[ \beta \mapsto \frac{e^{-c n R_n(\beta)} \, \pi(\beta)}{\int e^{-c n R_n(b)} \, \pi(b) \, db}, \]
where $\pi$ is a prior density for $\beta$ and $c > 0$ is a constant that needs to be carefully specified.  See \citet{martin.syring.chapter2022} for some general background on the Gibbs posterior distributions, which is also relevant to the developments here.  Some theory for Gibbs posteriors specifically in (Hedayat et al.'s version of) the MCID problem can be found in \cite{syring2020gibbs, syring.martin.mcid}. Computation of both the M-estimator $\hat\beta_n$ and the Gibbs posterior would be non-trivial, however, thanks to the non-smooth loss function.  One option, of course, is to replace the not-smooth loss with a smooth surrogate, which we discuss in Section~\ref{SS:smoothing} below.  But simply smoothing the loss function will not make sampling from the Gibbs posterior distribution any easier, so we take an altogether different approach in Section~\ref{S:bbqr}.  


\subsection{Loss function smoothing}
\label{SS:smoothing}

The presence of the sign function makes the original MCID loss function not smooth and, consequently, the empirical risk function is difficult to optimize.  This is not a major concern in the population MCID case, but it is more serious in the personalized MCID case thanks to the increased dimension.  As is often the case in computational problems, if the to-be-optimized function is unpleasant for some reason, then one can consider replacing it with something similar but also easier to manage. In our present MCID application, a surrogate loss function of the following form has been used, in \cite{hedayat2015minimum} and \cite{zhou2020interval}:
\[ \ell_\beta^{(\delta)}(x,y,z) = s_\delta\bigl( y(x - \beta^\top z) \bigr), \]
where $s_\delta(u) = \min\{(1 - u/\delta)_+, 1\}$, with $x_+ = \max(0, x)$ the positive part, and $\delta > 0$ a parameter that controls the smoothness.  Another slightly different form of the function $s_\delta$ has been used in \cite{zhou2020estimation} and \cite{zhou2021statistical}, in particular, 
\[ s_\delta(u) =\begin{cases}
    1       & u\leq 0\\
   1-2(u/\delta)^2  & 0< u\leq \delta/2\\
   2(1-u/\delta)^2 & \delta/2< u\leq \delta\\
   0 & u\geq\delta.
  \end{cases}
\]
As Figure~\ref{f:loss_function} makes clear, for small $\delta$, the surrogate loss functions are close to the original MCID loss function.  Moreover, solving the optimization problem defined by the surrogate loss function is much simpler, and the above authors make use of a representation of the surrogate loss in terms of a difference of convex functions, and they employ a version of the difference-of-convex-functions algorithm \citep{le1997solving}. 

\begin{figure}
    \centering
\scalebox{0.7}{\includegraphics{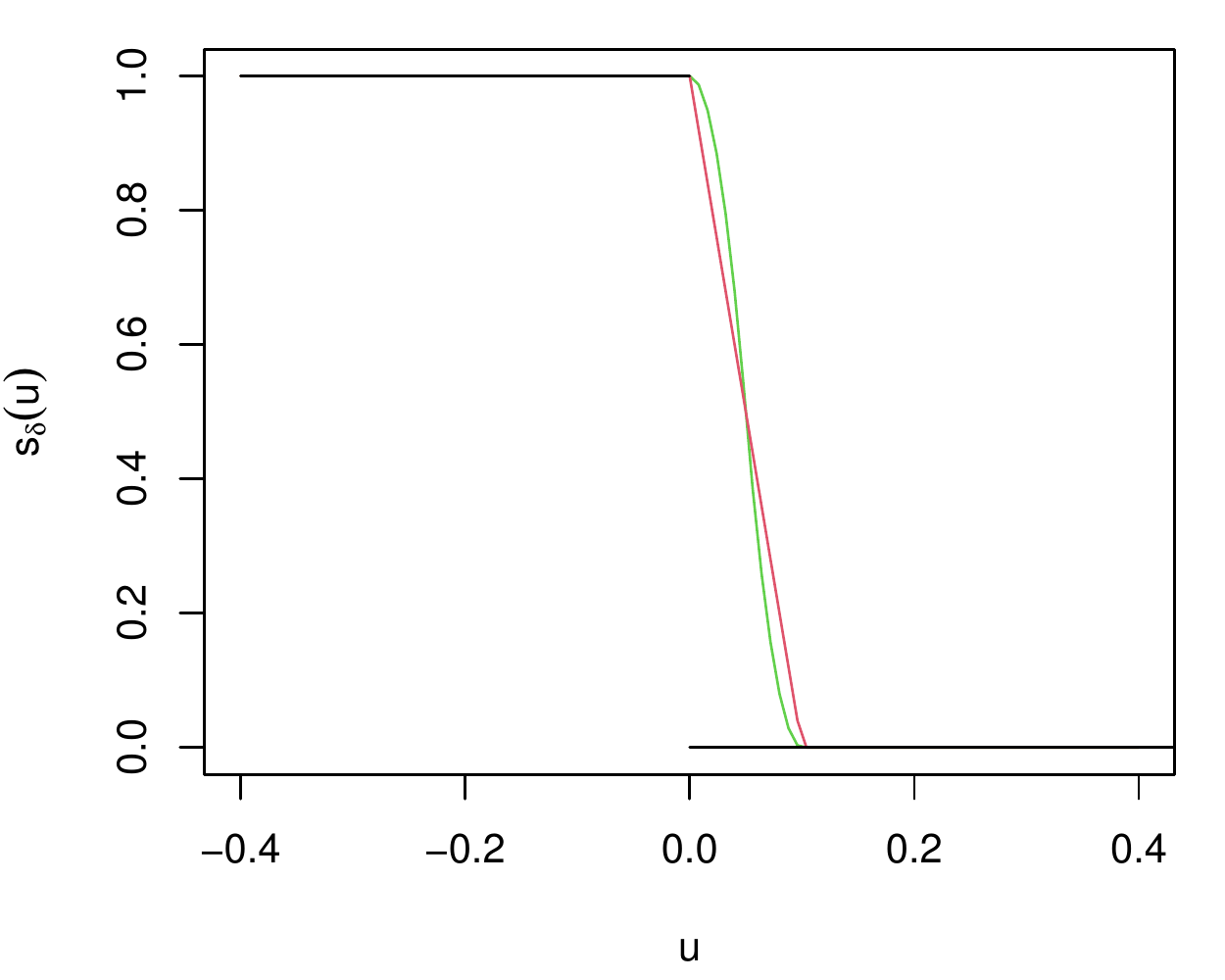}}
    \caption{Black line is the usual 0--1 loss, while the red and green lines are the smoothed loss functions $s_\delta$ in Hedayat et al.~and Zhou et al., respectively, with $\delta=0.1$.}
    \label{f:loss_function}
\end{figure}

The downside to the use of a surrogate loss is that it potentially affects the quality of inferences that can be made.  That is, when the loss is redefined to be smooth, the M-estimator would typically be asymptotically normal, but the center of the limit distribution would not be, in this case, $\beta^\star$, but some other value, say, $\beta_\delta^\dagger$.  Typically, as $\delta \to 0$, $\beta_\delta^\dagger \to \beta^\star$, but if one takes $\delta \to 0$, then there is a price one pays in terms of variance.  That is, when $\delta \to 0$, the bias goes down but the variance goes up; see \citet[][Sec.~2.3]{syring.martin.mcid} for more details.  \citet{zhou2020interval} side-step this challenge by showing, under certain conditions, that their smoothing strategy introduces no bias, {\em even when $\delta > 0$}, and they establish a corresponding asymptotic normality result.  We do, however, still have some concerns: in addition to some technical questions about the role the smoothing parameter plays in their asymptotic variance formula, we found that the coverage probability of their proposed asymptotic normality-based confidence interval falls well short of the nominal level in some of our simulations; see Section~\ref{SS:personalizedMCID_results}.  Therefore, we opt for an entirely different strategy to ensure our proposed generalized Bayes inference is calibrated, or reliable in a frequentist sense.

\section{MCID and binary quantile regression}
\label{S:bbqr}

\subsection{Smooth loss via a working model}

One generalization of the ordinary linear regression model is for binary responses.  A common example is probit regression which, in our present context, would look like  
\begin{align}
m_z(x) & := \prob(Y = +1 \mid X=x, Z=z) \notag \\
& = \Phi(x - \beta^\top z), \label{eq:probit}
\end{align}
where $\Phi$ is the standard normal distribution function.  The relevant computations needed to fit this model to data are most conveniently done via {\em data-augmentation}, the introduction of a continuous latent response variable  \citep[e.g.,][]{albert1993bayesian}.  Our experience is that this probit regression model relies too heavily on the Gaussian form in \eqref{eq:probit} and, therefore, is not ideal for inference on the MCID in general.  But the data-augmentation strategy will still prove to be useful.  

Another variation on classical linear regression model expresses the $\tau^\text{th}$ conditional quantile of the response, $\tau \in (0,1)$, as a function of covariates, rather than the conditional mean.  The original formulation of this quantile regression, in \cite{koenker1978regression}, was based on replacing squared-error loss with the so-called check-loss 
\[ \rho_\tau(u) = u(\tau - 1_{\{u < 0\}}), \quad u \in \RR, \quad \tau \in (0,1). \]

\cite{manski1985semiparametric} combined the above two ideas, where the ``$x - \beta^\top z$'' term represents the conditional median or, more generally, a conditional $\tau^\text{th}$ quantile of the latent response variable.  This is generally referred to as {\em binary quantile regression}. \cite{benoit2012binary} developed a Bayesian version that combines Manski's binary quantile formulation, the usual Bayesian approach to quantile regression \citep{yu2001bayesian} via the asymmetric Laplace distribution, and its data-augmentation version \citep{kozumi2011gibbs}.  The asymmetric Laplace distribution \citep{yu2005three} has distribution function 
\begin{align*}
F_\tau(u) & = \begin{cases}
\tau \exp\{-\rho_\tau(u)\} & \text{if $u \leq 0$} \\ 1 - (1-\tau) \exp\{-\rho_\tau(u)\} & \text{if $u > 0$}, 
\end{cases} 
\end{align*}
where $\tau \in (0,1)$ controls the skewness and $\rho_\tau$ is as above.  When $\tau=\frac12$, the asymmetric Laplace distribution simplifies to an ordinary/symmetric Laplace distribution.  This model can be augmented with a location and scale parameter as usual. 

In our personalized MCID context involving a binary outcome $Y$, a continuous diagnostic measure $X$, and a covariate $Z$, a binary quantile regression model assumes that $Y$ is determined by $\text{sign}(U)$, where the conditional distribution of $U$, given $X=x$ and $Z=z$, is asymmetric Laplace with skewness parameter $\tau$, location parameter $x-\beta^\top z$, and scale parameter $\eta > 0$ to be specified.  That is, the working model for the response $Y$, given $X$ and $Z$, assumes 
\begin{align}
g_\beta(x,z) := & \; \prob(Y=+1 \mid X=x, Z=z) \notag \\
= & \; \prob(U > 0 \mid \tau, X=x, Z=z) \notag \\ 
= & \; 1 - F_\tau\Bigl( \frac{\beta^\top z - x}{\eta} \Bigr) \notag \\
= & \; F_{1-\tau}\Bigl( \frac{x-\beta^\top z}{\eta} \Bigr) \notag \\
= & \, \begin{cases}
(1-\tau) e^{-\tau|x-\beta^\top z| / \eta} & \text{if $x \leq \beta^\top z$} \\
1 - \tau e^{-(1-\tau) |x-\beta^\top z| / \eta} & \text{if $x > \beta^\top z$}. 
\end{cases}
\label{eq:bqr}
\end{align}
Note that $x \mapsto g_\beta(x,z)$ is increasing for fixed $(\beta,z)$, which is consistent with the interpretation that larger values of the diagnostic measure increase the chance that the patient will feel that the treatment was effective.  We want to emphasize that we are not treating this as a genuine ``model'' for the data-generating process, e.g., we do not assume that there are true $\tau$ and $\eta$ values to be estimated.  Nevertheless, the above conditional probability determines a (possibly incorrectly specified) joint distribution for $(X,Y,Z)$, depending on the parameter $\beta$, with density 
\[ h_\beta(x,y,z) = g_\beta(x,z)^{(1+y)/2} \, \{1-g_\beta(x,z)\}^{(1-y)/2} \, p(x,z), \]
where $g_\beta(x,z)$ is as in \eqref{eq:bqr}, $p$ denotes the joint density of $(X,Z)$ supported on $\RR \times \ZZ$, with $\ZZ \subseteq \RR^q$, which does not depend on $\beta$ and can be effectively ignored, and the dependence throughout on $\eta$ and $\tau$ is implicit in the notation.  Again, we do not treat this as a proper ``model'' describing our beliefs about the data-generating process. We are primarily interested in whether 
\begin{equation}
\label{eq:l.bqr}
\ell_\beta^\text{\sc bqr}(x,y,z) := -\log h_\beta(x,y,z), 
\end{equation}
could be treated as a smoothed version of the MCID loss function.  The next subsection answers this question in the affirmative.

\subsection{Fisher consistency}
\label{SS:fisher}

As we show below, under the symmetry condition in Assumption~S below, the function $\ell_\beta^\text{\sc bqr}$ above does indeed serve as suitable smoothed version of the original MCID loss.

\begin{asmpS}
The conditional distribution of $X$, given $(Y,Z)$, has a density 
\begin{equation}
\label{eq:psi}
\psi_{y,z}(x) = \begin{cases} \psi(x - a_-^\top z) & \text{if $y=-1$} \\ \psi(x - a_+^\top z) & \text{if $y = +1$}, \end{cases} 
\end{equation}
where $\psi$ is a symmetric density strictly increasing on $(-\infty,0)$ and strictly decreasing on $(0,\infty)$, and the vectors $a_-$ and $a_+$ satisfy the component-wise inequality $a_+ > a_-$.  
\end{asmpS}

Under Assumption~S, \citet[][Lemma~1]{zhou2021statistical} establish that the ``true MCID,'' i.e., the unrestricted minimum of the risk function over all MCID functions $\theta$, is a linear function of the form $\theta^\star(z) = \beta^{\star \top}z$, where $\beta^\star = \frac12(a_- + a_+)$.  Therefore, this $\beta^\star$ effectively determines the personalized MCID, so we would want our risk minimizer to be equal to, or at least near, this $\beta^\star$ value.  Proposition~\ref{prop:consistent} below shows that, under the same condition,\footnote{Technically, the proof of Proposition~\ref{prop:consistent} only requires symmetry of $\psi$ in \eqref{eq:psi}, the monotonicity conditions are not needed. We formulate the result using Assumption~S to maintain the connection to the expression for $\beta^\star$ given in \citet{zhou2021statistical} under those conditions.} the risk function 
\begin{equation}
\label{eq:R.bqr}
R_\eta^\text{\sc bqr}(\beta) := \E\{\ell_\beta^\text{\sc bqr}(X,Y,Z)\}, 
\end{equation}
with $\ell_\beta^\text{\sc bqr}$ defined in \eqref{eq:l.bqr}, has the same minimizer, $\beta^\star$, as Zhou et al.'s risk function, at least as $\eta \to 0$.  In other words, under Assumption~S, the risk $R_\eta^\text{\sc bqr}$ based on the BQR-driven surrogate loss is {\em Fisher consistent} at $\beta^\star$ \citep[e.g.,][p.~287]{cox.hinkley.book}. 

\begin{prop}
\label{prop:consistent}
Let $\psi_{y,z}$ denote the conditional density of $X$, given $Y$ and $Z$, which need not satisfy the symmetry conditions in Assumption~S.  As $\eta \to 0$, the minimizer of $R_\eta^\text{\sc bqr}(\beta)$ is the solution to the equation 
\begin{equation}
\label{eq:mess}
\tau \varpi \int_\ZZ z \int_{-\infty}^{\beta^\top z} \psi_{+1, z}(x) \, dx \, dz = (1-\tau) (1-\varpi) \int_\ZZ z \int_{\beta^\top z}^\infty \psi_{-1,z}(x) \, dx \, dz. 
\end{equation}
In particular, if Assumption~S holds and $\tau = 1-\varpi$, then $R_\eta^\text{\sc bqr}$ is Fisher consistent at $\beta^\star$ as $\eta \to 0$.  If $\tau=1-\varpi=\frac12$, then $R_\eta^\text{\sc bqr}$ is Fisher consistent at $\beta^\star$ for each $\eta > 0$. 
\end{prop}

\begin{proof}
See Appendix~\ref{S:proof}.
\end{proof}

Next we give a brief illustration of the result in Proposition~\ref{prop:consistent}.  For simplicity and ease of visualization, we focus on the population MCID case with no covariate $Z$ and a scalar MCID $\theta$.  Suppose that the conditional density of $X$, given $Y=y$, is $\psi_y(x) = \nm(x \mid y, 2^2)$, a normal density with mean $\pm 1$, depending on $y \in \{-1,+1\}$, and standard deviation 2.  We consider two cases: $\varpi = 0.7$ and $\varpi = 0.5$.  Note that the latter is the special case highlighted in Proposition~\ref{prop:consistent}, where Hedayat et al.'s and Zhou et al.'s definitions of the MCID coincide and our risk function is exactly Fisher consistent at this common MCID value.  Having specified the marginal distribution of $Y$ and the conditional distribution of $X$, given $Y$, we have all we need to evaluate the various risk functions.  Figure~\ref{fig:mcid.risk.comp} plots the (scaled\footnote{The risk functions have different scales, so we divide by their respective minima to force the minimum value to be 1 in the plot. This rescaling does not affect the risk function's shape or minimizer.}) risk functions corresponding to Hedayat et al.'s loss, Zhou et al.'s loss, and our BQR-based surrogate loss with $\eta=0.1$.  Panel~(a) corresponds to $\tau=1-\varpi=0.3$ so we see that ours and Zhou et al.'s risk functions are minimized at $\theta=0$ while, as expected, Hedayat et al.'s risk is minimized at a different location, roughly $\theta=-1.8$. Panel~(b) corresponds to the case $\tau=1-\varpi=0.5$ and, as predicted by Proposition~\ref{prop:consistent}, all three risk functions share the same minimizer.  

To follow up on the previous illustration, in Figure~\ref{fig:mcid.erisk.comp}, we plot the empirical risk functions based on the three losses discussed above, along with smoothed versions of Hedayat et al.'s and Zhou et al.'s.  These results are based on a sample of size $n=250$ from the two joint distributions of $(X,Y)$ described above.  As expected, the behavior here mirrors that in Figure~\ref{fig:mcid.risk.comp}, but there are some additional insights to be gleaned here.  The basic empirical risk functions for the Hedayat et al.~and Zhou et al.~formulations are clearly very rough and would be virtually impossible to optimize if this were a higher-dimensional personalized MCID scenario.  The motivation for smoothing is to simplify the optimization and, as the figure shows, the smoothed versions are indeed smoother and, consequently, easier to optimize.  Remarkably, the BQR-based empirical risk points in the right direction and is very smooth, very easy to optimize.  Moreover, as we show below, the real advantage of our proposed formulation of the loss in \eqref{eq:l.bqr} is that it leads to relatively straightforward generalized posterior computations.

\begin{figure}[t]
\begin{center}
\subfigure[$\tau = 1-\varpi = 0.3$]{\scalebox{0.6}{\includegraphics{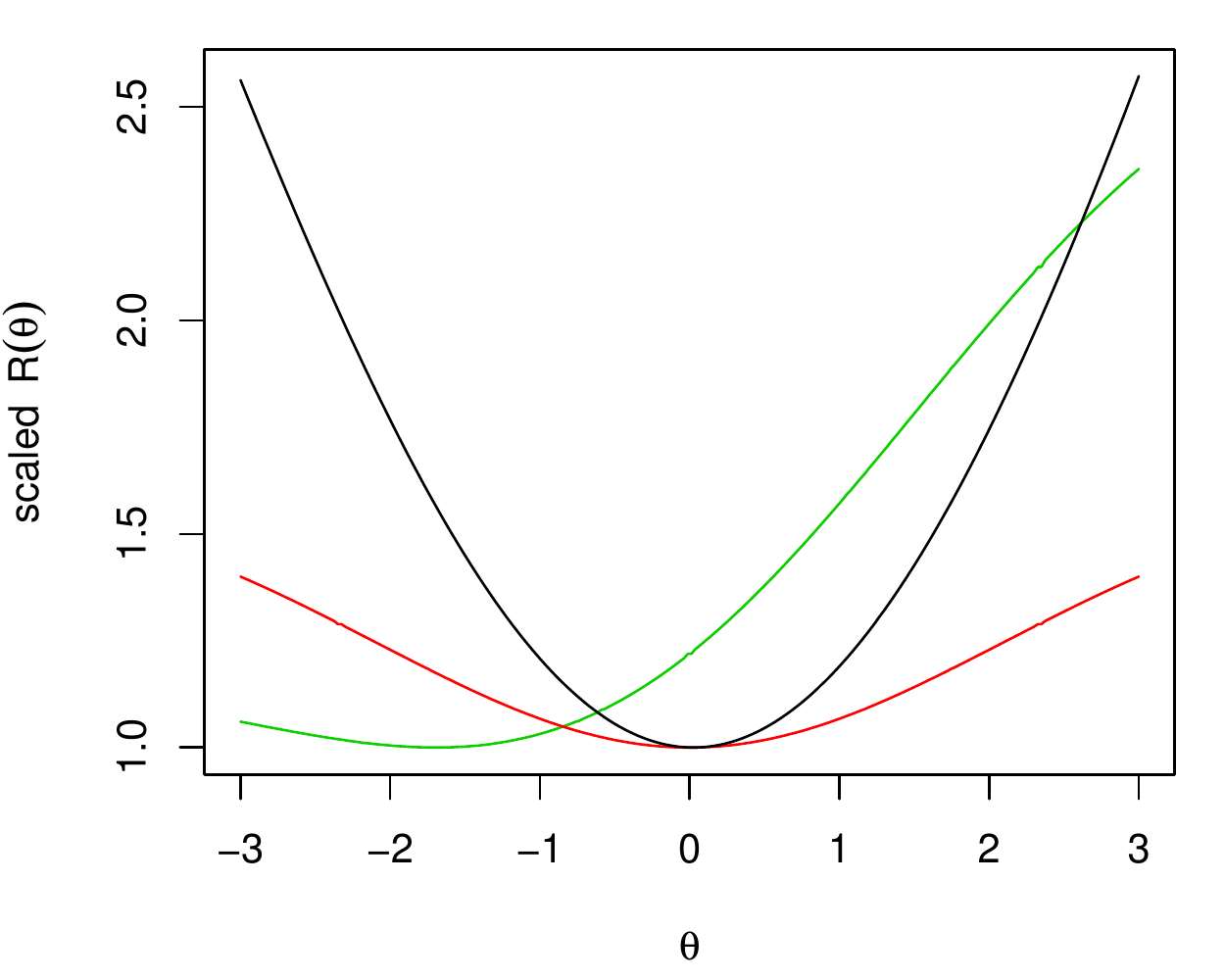}}}
\subfigure[$\tau = 1-\varpi = 0.5$]{\scalebox{0.6}{\includegraphics{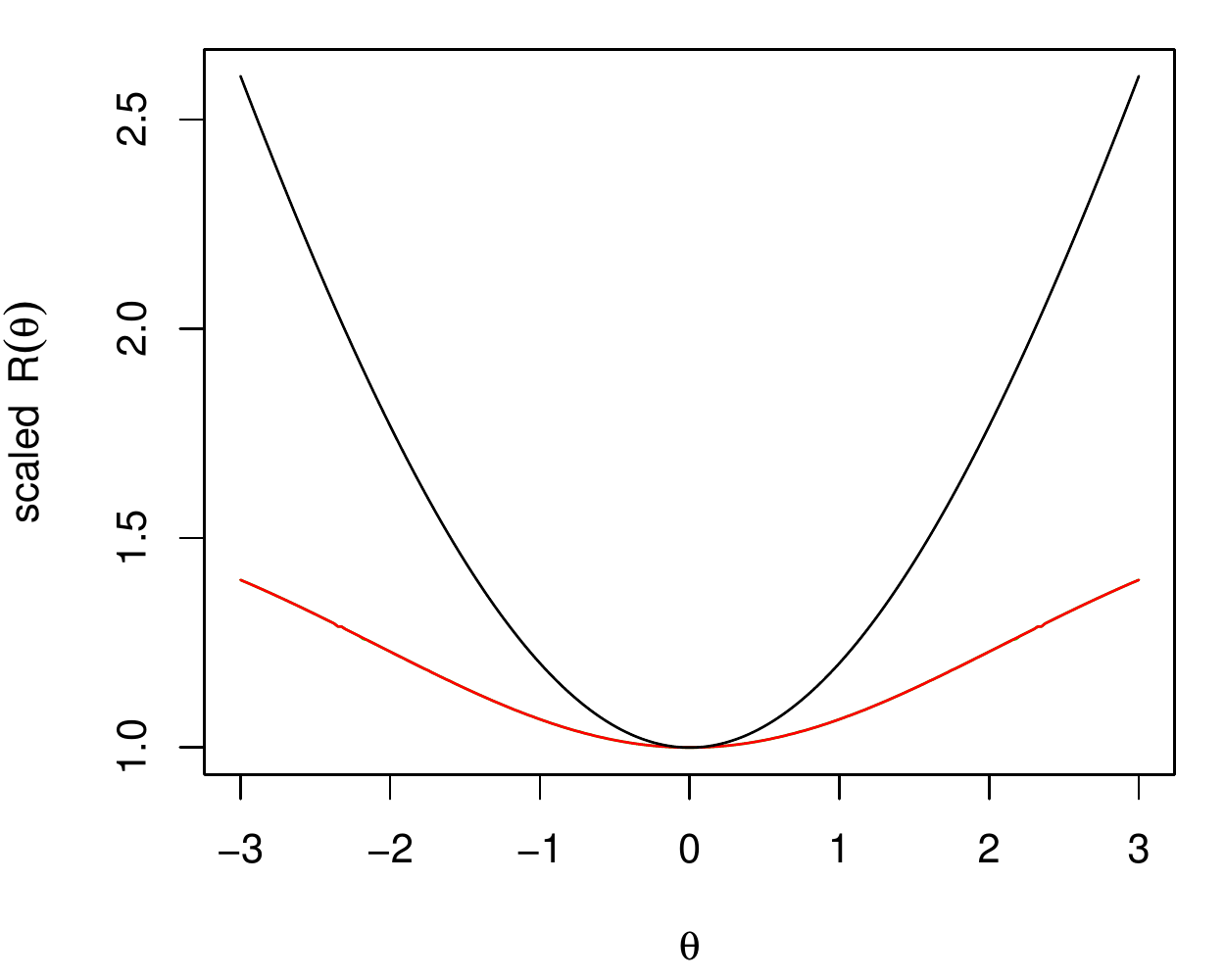}}}
\end{center}
\caption{Plots of the (scaled) risk functions for the two different $\varpi$ values as described in the main text. Red, green, and black lines are the risks based on Zhou et al.'s loss, Hedayat et al.'s loss, and our BQR-based loss with $\eta=0.1$. In Panel~(b), the green line is underneath the red line.}
\label{fig:mcid.risk.comp}
\end{figure}

\begin{figure}[t]
\begin{center}
\subfigure[$\tau = 1-\varpi = 0.3$]{\scalebox{0.6}{\includegraphics{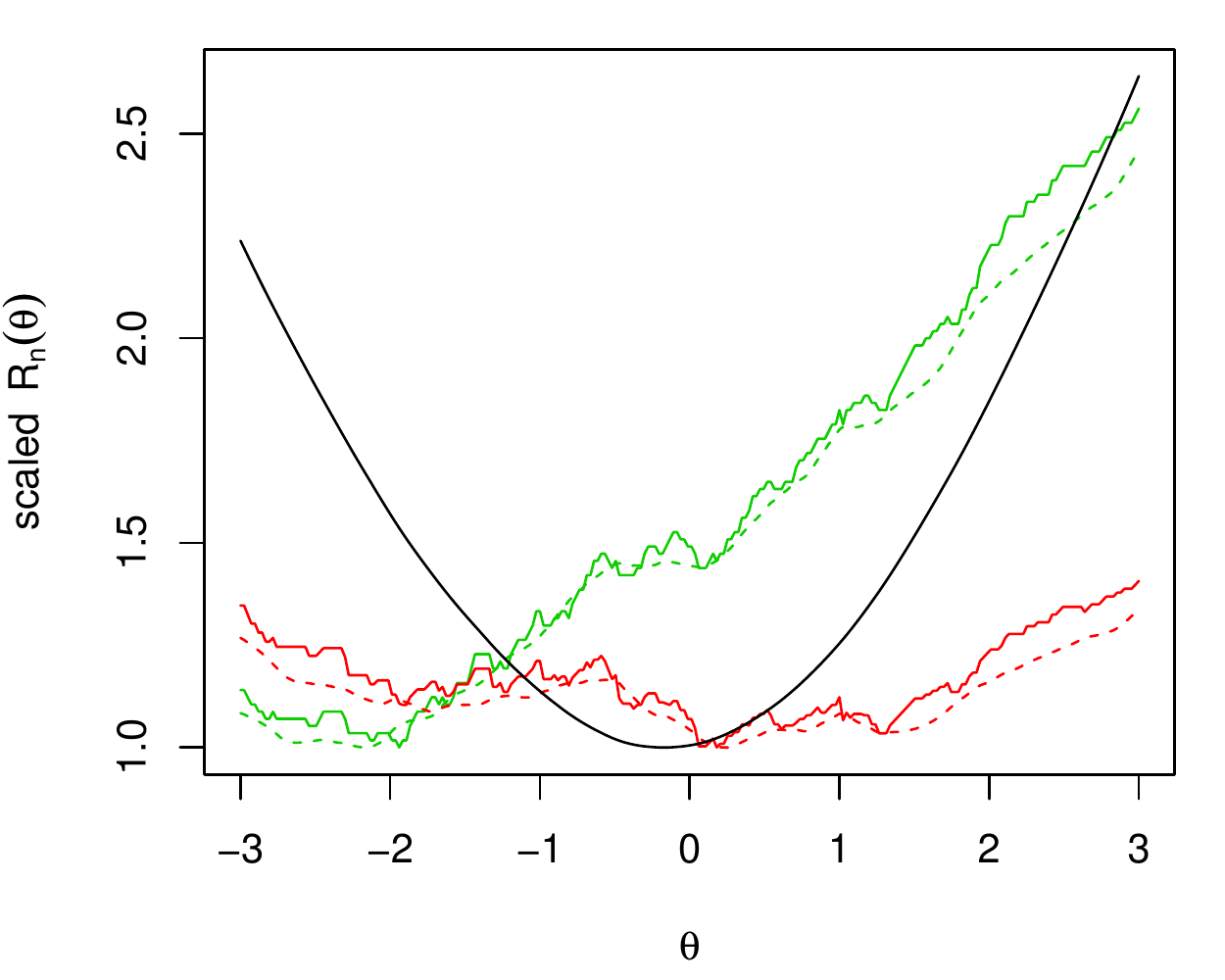}}}
\subfigure[$\tau = 1-\varpi = 0.5$]{\scalebox{0.6}{\includegraphics{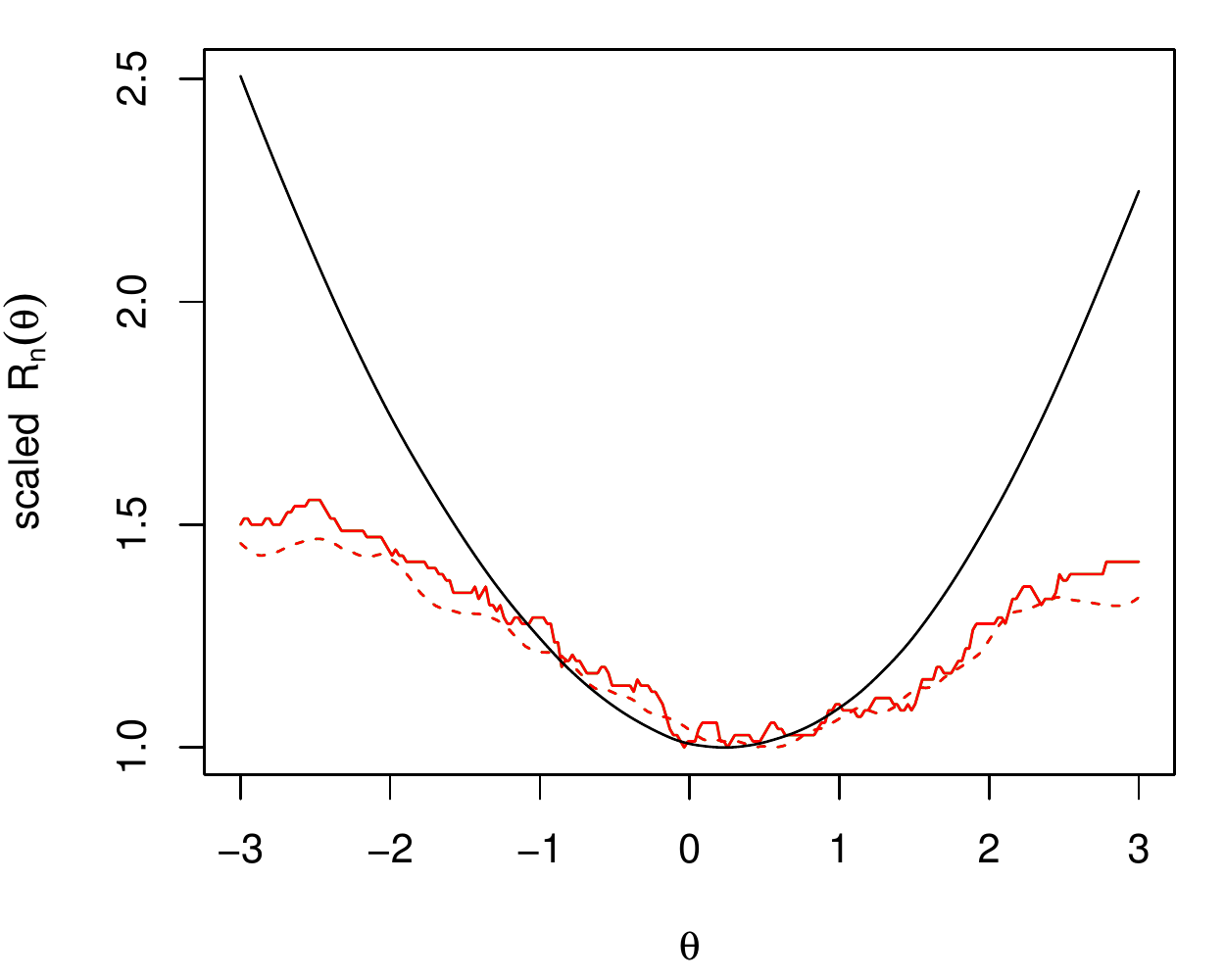}}}
\end{center}
\caption{Plots of the (scaled) empirical risk functions for the two different $\varpi$ values as described in the main text. Red, green, and black lines are the risks based on Zhou et al.'s loss, Hedayat et al.'s loss, and our BQR-based loss with $\eta=0.1$. Dashed lines correspond to the smoothed versions, with $\delta=0.25$.}
\label{fig:mcid.erisk.comp}
\end{figure}

One additional remark concerning Proposition~\ref{prop:consistent} is in order.  In particular, we only establish Fisher consistency when $\tau$ depends on the probability $\varpi := \prob(Y=+1)$.  But $\varpi$ is unknown in applications, so we cannot expect to have $\tau = 1-\varpi$ exactly.  We can, of course, get an empirical estimate $\hat\varpi$ based on the observed data, and then take $\tau = 1-\hat\varpi$.  Zhou et al.'s solution depends on $\varpi$ too, so they also recommend using a plug-in estimator.  When it is not possible to take $\tau = 1-\varpi$ exactly, the characterization of the solution in \eqref{eq:mess} can be helpful for investigating the bias this introduces.

\subsection{Latent variable formulation}
\label{SS:latent}

An advantage to the proposed binary quantile regression-based loss function is that it is smooth, as is clear from the plots presented in the previous subsection.  A further advantage is that the underlying binary quantile regression model admits a latent variable formulation which facilitates efficient computation.  In particular, this representation can be used to develop Markov chain Monte Carlo algorithms for (generalized) posterior sampling as is our focus here in the present MCID application.

The Bayesian quantile regression model was first developed and implemented in \cite{benoit2012binary} by assuming an asymmetric Laplace distribution latent variables and a Metropolis-within-Gibbs sampling algorithm for posterior inference. But these computations can still be expensive. The various computational challenges were overcome in \cite{mollica2017bayesian} by re-expressing the standard asymmetric Laplace distribution error as a mixture of exponential and Gaussian distributions.  That is, in our formulation, we have 
\[ U = x - \beta^\top z + \eta \eps, \]
where the error $\eps$ has an asymmetric Laplace distribution with skewness parameter $\tau$.  Mollica and Petrella observed that $\eps$ could be re-expressed as
\[ \eps = c_{1} V_{1} + c_{2} V_{1}^{1/2} V_{2}, \]
where the coefficients are $c_1=\frac{1-2\tau}{\tau(1-\tau)}$ and $c_2^2=\frac{2}{\tau(1-\tau)}$, and $V_1$ and $V_2$ are independent random variables with $V_{1} \sim \mathsf{Exp}(1)$ and $V_{2} \sim \mathsf{N}(0,1)$.  If it happens that $\tau=\frac12$, then $c_1=0$ and the above simplifies a bit.  

Then the key idea is as follows.  We can write a complete-data likelihood for $\beta$ and the latent variables $(U_i, V_{1i})$ given the observations $(X_i, Y_i, Z_i)$ as 
\begin{align*}
L_n(\beta, u, v_1) & \propto \prod_{i=1}^n \bigl[ \{(1+y_i) \, 1_{\{u_i\geq 0\}}+(1-y_i) \, 1_{\{u_i<0\}} \} \\
& \qquad \qquad \times \nm\bigl(u_i \mid \text{mean}_i(\eta), \text{var}_i(\eta) \bigr) \, {\sf Exp}(v_{1i} \mid 1) \bigr], 
\end{align*}
where, for $i=1,\ldots,n$, 
\begin{align*}
\text{mean}_i(\eta) & = (x_i-\beta^\top z_i)+\eta c_1 v_{1i} \\
\text{var}_i(\eta) & = \eta^2 c_2^2 v_{1i},
\end{align*}
and the dependence on $\eta$ and $\tau$ is mostly implicit in the notation.  Having this complete-data likelihood for the parameter and latent variables creates an opportunity to design an EM algorithm for maximizing the likelihood for $\beta$ in the working model, i.e., finding our proposed empirical risk minimizer, and to develop a Gibbs sampler for generalized posterior inference.  We focus on the latter in Section~\ref{S:general.post} below.


\section{A generalized posterior for the MCID}
\label{S:general.post}

\subsection{Definition}

We have shown that our surrogate loss function, based on a purposely misspecified binary quantile regression model, admits a risk function that, under certain conditions, roughly centers around the MCID-determining value $\beta^\star$ defined above. Here we present the details of our proposed generalized posterior for inference on the personalized MCID and, in particular, the Gibbs sampling algorithm we use to compute the posterior. 

To start, the definition of our generalized posterior distribution is simple to state.  Given a prior distribution for $\beta$ with density $\pi$ supported on $\RR^q$, define 
\[ \pi_n^{(\eta)}(\beta) \propto \exp\{-n \widehat R_\eta^\text{\sc bqr}(\beta)\} \, \pi(\beta), \quad \beta \in \RR^q, \]
where $\widehat R_\eta^\text{\sc bqr}$ is the $\eta$-dependent empirical risk function corresponding to the risk function $R_\eta^\text{\sc bqr}$ in \eqref{eq:R.bqr}, i.e., 
\[ \widehat R_\eta^\text{\sc bqr}(\beta) = \frac1n \sum_{i=1}^n \ell_\beta^\text{\sc bqr}(X_i, Y_i, Z_i), \]
where $\ell_\beta^\text{\sc bqr}$ is the loss defined in \eqref{eq:l.bqr}.  The dependence of $\pi_n^{(\eta)}$ on $\tau$ is implicit.  

This is different from the Gibbs posterior distribution formulation taken in \citet{syring2020gibbs} because here we are not directly using the loss function that defines the MCID, we have a surrogate loss based on a purposely-misspecified model.  Also, the role played by $\eta$ here is different than that of the learning rate parameter in \citet{syring2020gibbs}, though we will handle its data-driven selection in a similar manner; see Section~\ref{SS:pmcid.gpc} below.  The key point, the reason for considering this alternative formulation in the first place, is that the original loss functions do not admit generalized posterior distributions that can be efficiently sampled from.  Here we take advantage of the latent variable formulation discussed in Section~\ref{SS:latent} to design an efficient Gibbs sampler for generating realizations from the generalized posterior $\pi_n^{(\eta)}$ above.

\subsection{Gibbs sampler}
\label{SS:mcid.details}

The idea is to start with a posterior distribution determined by the complete-data likelihood presented in Section~\ref{SS:latent}, the one that involves both the parameter $\beta$ and the latent variables $(U_i, V_{1i})$, $i=1,\ldots,n$.  Then the generalized posterior for $\beta$ will drop out from the corresponding joint posterior for $(\beta,U,V_1)$ by the usual marginalization process.  

Here we suppose that $\beta$ has a $\nm(\mu_0, \Sigma_0)$ prior distribution.  Then the joint posterior density for $(\beta, V, V_1)$ is given by 
\begin{align*}
\pi_n^{(\eta)}(\beta, u, v_1) & \propto \exp\bigl\{-\tfrac{1}{2}(\beta -\mu_0)^\top \Sigma_0\Inv(\beta-\mu_0) \bigr\} \\
& \qquad \times \frac{1}{\prod_{i=1}^n\sqrt{v_{1i}}} \exp{\left\{-\frac{\sum_{i=1}^n (u_i-\text{mean}_i(\eta))^2}{2 \text{var}_i(\eta)}\right\}} \\
& \qquad \times\exp\Bigl(-\sum_{i=1}^nv_{1i} \Bigr) \prod_{i=1}^n \bigl[ (1+y_i) \, 1_{\{u_i\geq 0\}} + (1-y_i) \, 1_{\{u_i<0\}}\bigl].
\end{align*}
From this joint probability model for $(\beta, u, v)$, it is possible to write down the full conditionals and, in turn, develop a Gibbs sampling algorithm.  Indeed, following \cite{mollica2017bayesian}, we have the following full conditionals:
\begin{align*}
(u_i \mid v_{1i},x_i,y_i,\beta) & \sim \begin{cases}
    \mathsf{trN}_{[0,\infty)}(\text{mean}_i(\eta), \text{var}_i(\eta)) & y_i=+1\\
   \mathsf{trN}_{(-\infty,0)}(\text{mean}_i(\eta), \text{var}_i(\eta)) & y_i=-1,
  \end{cases} \\
(v_{1i} \mid u_i,x_i,y_i,\beta) & \sim \mathsf{GIG}(\tfrac12, \, \{u_i-(x_i-\beta^\top z_i)\}^2/\eta^2 c_2^2, \, 2+c_1^2/c_2^2), \\
(\beta \mid u_i,v_{1i},y_i,x_i) & \sim \mathsf{N}_p(\mu_n,\Sigma_n),
\end{align*}
where ${\sf trN}_A(m,s^2)$ denotes a normal distribution with mean $m$ and variance $s^2$ truncated to the interval $A$, 
\begin{align*}
\Sigma_n\Inv & =\eta^{-2} c_2^{-2}Z^\top \Lambda\Inv Z+\Sigma_0\Inv \\
\mu_n & = \Sigma_n\{\Sigma_0\Inv\mu_0-\eta^{-2} c_2^{-2} Z\Tra \Lambda\Inv(u-\eta c_1 v_1-X)\} \\
\Lambda & = \text{diag}(v_{11},\ldots, v_{1n}), 
\end{align*}
and $\mathsf{GIG}(\nu,a,b)$ denotes the generalized inverse Gaussian distribution \citep{jorgensen2012statistical} with density function given by 
\[ x \mapsto \frac{(b/a)^{\nu/2}}{2\kappa_\nu(\sqrt{ab})}x^{\nu-1}\exp{\Bigl(-\frac{ax\Inv+bx}{2}\Bigr)}, \quad x>0, \]
where $\kappa_\nu$ is the modified Bessel function \citep[e.g.,][]{watson1995treatise}.  Recall that $c_1$ and $c_2$ implicitly capture the generalized posterior distribution's dependence on $\tau$.  

Then the Gibbs sampler algorithm proceeds by iteratively sampling from these full conditionals, which results in a collection of samples 
\[ (\beta^{(m)}, u^{(m)}, v_1^{(m)}), \quad m=1,\ldots,M, \]
where $M$ is the designated Monte Carlo sample size.  Our interest is in the corresponding marginal posterior distribution for the MCID-determining coefficient vector $\beta$, so we simply ignore the $(u,v_1)$ components, leaving us with 
\begin{equation}
\label{eq:beta.m}
\{\beta^{(m)}: m=1,\ldots,M\} 
\end{equation}
as an approximate sample from the Gibbs posterior $\pi_n^{(\eta)}$.  With this we can proceed to make inference on $\beta$ or on the MCID $\beta^\top z_{\text{new}}$ for some relevant value $z_\text{new}$ of $Z$.  

\subsection{Calibration-focused choice of $\eta$}
\label{SS:pmcid.gpc}

As we showed above, our proposed binary quantile regression-based loss function, $\ell_\beta^\text{\sc bqr}$ leads to risk function that is minimized in roughly the right place.  However, we readily acknowledge that this is not a correctly specified model and, therefore, some appropriate tuning of the scale/learning rate parameter $\eta$ is necessary in order to have calibrated generalized Bayes posterior inferences.  Here we follow the {\em generalized posterior calibration} (GPC) strategy in \cite{syring2019calibrating} to tune the learning rate in such a way that calibration is achieved.  In this present application, the generalized posterior's dependence on the tuning parameter is more complicated than in previous applications of GPC \citep[e.g.,][]{wang.martin.auc,  gibbs.quantile, syring.martin.mcid, wu2020comparison}.  This highlights the GPC algorithm's versatility.  

Recall that we have data $D^n = (D_1,\ldots,D_n)$, where each $D_i = (X_i, Y_i, Z_i)$ consists of a diagnostic measure, a binary patient-reported outcome, and a patient profile/covariate, respectively.  Now suppose that we have another patient profile, $\tilde z$, and that the goal is inference on $\tilde\theta = \beta^\top \tilde z$, the MCID for a patient having profile $\tilde z$.  
Let $\widetilde\Pi_n^{(\eta)}$ denote the marginal posterior for $\tilde\theta=\beta^\top \tilde z$ based on the generalized posterior for $\beta$ and the given $\tilde z$.  That is, for the samples $\beta^{(m)}$ in \eqref{eq:beta.m}, we get 
\[ \tilde\theta^{(m)} = \beta^{(m)\top} \tilde z, \quad m=1,\ldots,M. \]
Of course, this distribution depends on $D^n$, $\eta$, and $\tilde z$, as well as the other prior inputs.  Fix a significance level $\alpha \in (0,1)$ and let $C_\alpha^{(\eta)}(D^n)$ denote the corresponding $100(1-\alpha)$\% posterior credible interval for $\tilde\theta$.  The goal of the GPC strategy is to choose $\eta$ such that the $100(1-\alpha)$\% generalized posterior credible interval is also (at least approximately) a $100(1-\alpha)$\% confidence interval.  That is, GPC aims to choose $\eta$ to solve the equation 
\begin{equation}
\label{eq:cvg}
c_\alpha(\eta) = 1-\alpha, 
\end{equation}
where 
\[ c_\alpha(\eta) = P\{C_\alpha^{(\eta)}(D^n) \ni \beta^{\star \top} z\} \]
is the coverage probability of the aforementioned credible intervals.  A relevant question is if a solution to the above equation even exists.  In previous applications of the GPC, it was clear that the generalized posterior was centering around an empirical risk minimizer that did not directly depend on the choice of $\eta$, so that adjusting $\eta$ largely only affected the spread of the posterior.  In such cases, it is easy to convince oneself that at least an approximate solution can be found.  In the present case, however, we technically cannot rule out the possibility that the minimizer of $R_\eta^\text{\sc bqr}$ depends on $\eta$, which suggests that $\eta$ might affect both the center and spread of our generalized posterior distribution.  As we show in Appendix~\ref{S:more}, however, in our experience, the posterior center depends only mildly on $\eta$, so we have every reason to expect that (a slight variation on) the GPC strategy can at least approximately solve the equation \eqref{eq:cvg}.  

Our variation on the GPC strategy proceeds as follows.  The idea is to bootstrap to approximate the coverage probability function $c_\alpha$, and then use a stochastic approximation scheme to solve \eqref{eq:cvg}.  For $b=1,\ldots,B$, let $D_b^n$ denote an iid sample of size $n$ from the empirical distribution of $D^n$, i.e., a sample with replacement from the data points in $D^n$. At a given $\eta$, the coverage probability can be approximated by 
\[ \hat c_\alpha(\eta) = \frac1B \sum_{b=1}^B 1\{ C_\alpha^{(\eta)}(D_b^n) \ni \hat\beta_\eta^\top \tilde z\}, \]
where $\hat\beta_\eta$ denotes the ($\eta$-dependent) generalized posterior mean based on the original data $D^n$.  To solve the aforementioned equation, this process needs to be repeated at several strategically-chosen $\eta$ values.  We use the Robbins--Monro stochastic approximation scheme to adaptively choose those $\eta$ values according to the updates 
\[ \eta_{t+1} = \eta_t + k_t \{ \hat c_\alpha(\eta_t) - (1-\alpha)\}, \quad t \geq 0, \]
where $\eta_0$ is some initial guess.  The stochastic approximation updates in this GPC algorithm are terminated when some pre-specified convergence criterion is met, at which point a value $\hat\eta$ is returned.  Then our inferences about the $\tilde z$-dependent MCID, $\tilde\theta=\beta^\top \tilde z$, are based on posterior samples from $\widetilde\Pi_n^{(\hat\eta)}$.  

Computationally, this can be expensive because each evaluation of the posterior credible interval requires a run of the data-augmentation-based Gibbs sampler described in Section~\ref{SS:mcid.details}. Indeed, the Gibbs sampler approach on binary quantile regression model requires generating $n$ latent variables, and the GPC algorithm requires these evaluations for each bootstrap sample and at each iteration $t$ in the Robbins--Monro updates. To save on computation time, our implementation uses the {\tt Rcpp} package in R. It is also possible to do the $B$-many posterior computations at a given $\eta$ in parallel.

\section{Numerical examples}
\label{SS:personalizedMCID_results}
\subsection{Models, methods, and metrics}

Here we consider two pairs of simulation examples.  Since our primary goal is to compare the performance of our method with that of Zhou et al., we use their same examples.  It happens that these examples all consider balanced cases, i.e., where $\varpi = \frac12$, so the target MCID is the ``MCID'' as defined by Zhou et al.~and by Hedayat et al.

The first pair of examples considered here are those used for illustration in \cite{zhou2021statistical}.  In both cases, the marginal distribution of $Y$ is given by 
\[ Y \sim 2 \, {\sf Ber}(0.5) - 1. \]
They differ in the specification of the marginal distribution for $Z$ and conditional distribution of $X$, given $(Y,Z)$.  

\begin{ex}
The marginal distribution for $Z$ is 
\[ Z \sim \nm(1, 0.1^2), \]
and the conditional distribution of $X$, given $(Y,Z)$ is 
\[ (X \mid Y,Z) \sim \begin{cases}
    \nm(a_0 + a_1 Z, 0.1^2)       & Y=+1\\
   \nm(b_0 + b_1 Z, 0.1^2)  & Y=-1,\\
  \end{cases}
\]
where $a=(a_0, a_1) = (0.1,0.55)$ and $b=(b_0, b_1) =(-0.1,0.45)$.  
\end{ex}

\begin{ex}
The marginal distribution for $Z$ is 
\[ Z \sim \nm_2(1_2, 0.1^2 \, I_2), \]
bivariate normal, and the conditional distribution of $X$, given $(Y,Z)$ is
\[ (X \mid Y,Z) \sim\begin{cases}
    \nm(a_0 + a_1 Z_1 + a_2 Z_2, 1)       & Y=+1\\
   \nm(b_0 + b_1 Z_1 + b_2 Z_2, 1)  & Y=-1,\\
  \end{cases}
\]
where $a=(a_0, a_1, a_2) = (0.05,0.55,1.05)$ and $b=(b_0, b_1, b_2) =(-0.05,0.49,0.95)$. 
\end{ex}

The next two examples are exactly those used for illustration in \cite{zhou2020interval}.  In both cases, the data $(X,Y,Z)$ has $Y$ marginal given by 
\[ (Y \mid X, Z) \sim 2 \, {\sf Ber}\{ F_{X|Z}(x \mid z)\} - 1, \]
where $F_{X|Z}$ denotes the conditional distribution function of $X$, given $Z$.  The marginal distribution for $Z$ is bivariate normal, $Z \sim \nm_2(0_2, I_2)$. What differs between the two examples is $F_{X|Z}$.

\begin{ex}
The conditional distribution of $X$, given $Z$ is  
\[ (X \mid Z) \sim \nm(\beta_0 + \beta_1 Z_1 + \beta_2 Z_2, 1)\]
where $a=(a_0, a_1, a_2) = (0,1,2)$.
\end{ex}

\begin{ex}
The conditional distribution of $X$, given $Z$ is  
\[ (X \mid Z) \sim \nm(\beta_0 + \beta_1 Z_1 + \beta_2 Z_2- \beta_1 Z^2_1 - \beta_2 Z^2_2, 1)\]
where $(a_0, a_1, a_2) = (0,1,2)$.
\end{ex}

Suppose that we have a new patient profile $\tilde z$, the true MCID in the four examples is $\beta^{\star \top} \tilde z$. Our goal here is to investigate the performance of different methods on estimating this MCID.  Here we compare our proposed method---generalized posterior based on binary quantile regression + GPC---with the method proposed in \cite{zhou2021statistical}. 

The metrics we base these comparisons on are 
\begin{itemize}
\item coverage probability of the 95\% credible interval for the new patient's MCID;
\vspace{-2mm}
\item bias of the new patient's MCID;
\vspace{-2mm}
\item mean square error of the new patient's MCID.
\end{itemize} 
The metrics are evaluated by averaging across $R=250$ replications, each with sample size $n=200$. For the \cite{zhou2021statistical}, we follow the recommendations in the {\tt MCID} package in R \citep{R.mcid} 
and set the grid points for $\lambda=10^{(-3+0.1i)}$, for $i=1,\dots,60$, and $\delta \in \{0.1,0.2,0.3\}$. Zhou et al.'s tuning parameters, $\lambda$ and $\delta$, are selected using cross-validation.  For our proposed binary quantile regression + GPC method, we proceed by introducing a Gaussian prior for $\beta$, which is $\nm_p(0,I_p)$. Our expectation is that Zhou et al.'s method will be able to estimate the MCID parameter with small bias, but might not have a satisfactory coverage probability. 

\subsection{Results}

A first question is: what should we expect $\eta$ to be?  Since this is a simulation study, we can simulate from the true data-generating process to find an ``oracle'' learning rate value, $\eta^\star$, i.e., the value that would make the true coverage probability equal to $1-\alpha=0.95$.  We performed this simulation for Examples~1--4 and the results are displayed in Figure~\ref{fig:MCID_lr}.  The plot shows (a Monte Carlo estimate of) the coverage probability of our 95\% generalized posterior credible intervals as a function of $\eta$.  In all four cases, the coverage probability is increasing in $\eta$, which is to be expected.  We also see that there exists an $\eta^\star$ value---where the black and (solid) red lines intersect---that makes the coverage probability roughly 0.95 in each case.  So our hope is that GPC will select a learning rate value that tends to be near these oracle values in the respective examples; if so, then we would expect the generalized posterior distribution for the MCID, based on our proposed BQR + GPC method, to be approximately calibrated.   The $(x,y)$ coordinates of the red ``X'' in these plots represents the average $\hat\eta$ value chosen by GPC and the empirical coverage probability across replications.  This demonstrates that GPC is selecting $\hat\eta$ values sufficiently near the oracle $\eta^\star$, and that our corresponding generalized posterior inference for the MCID is calibrated, or at least approximately so.   

\begin{figure}[t]
\begin{center}
\subfigure[Example 1]{\scalebox{0.5}{\includegraphics{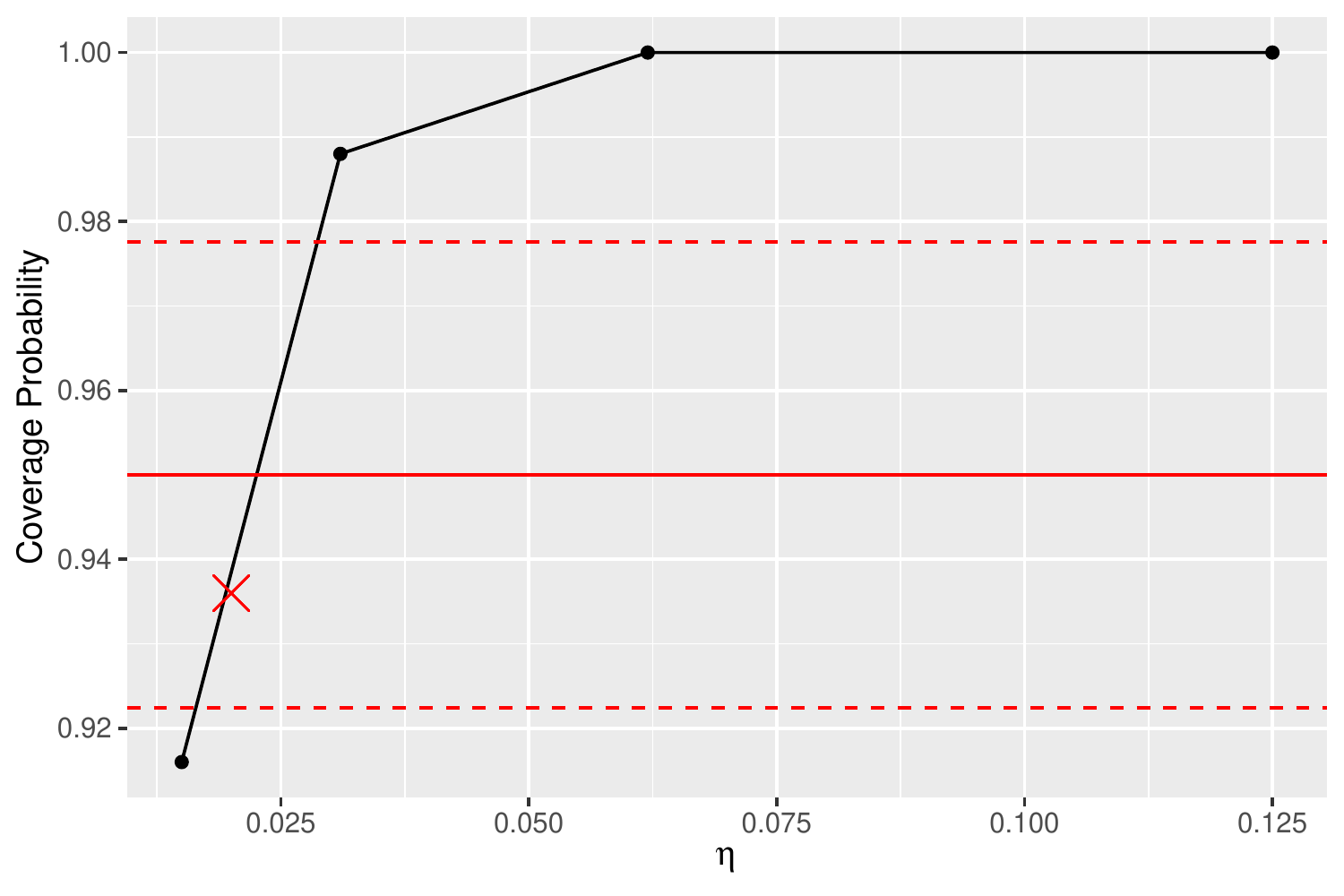}}}
\subfigure[Example 2]{\scalebox{0.5}{\includegraphics{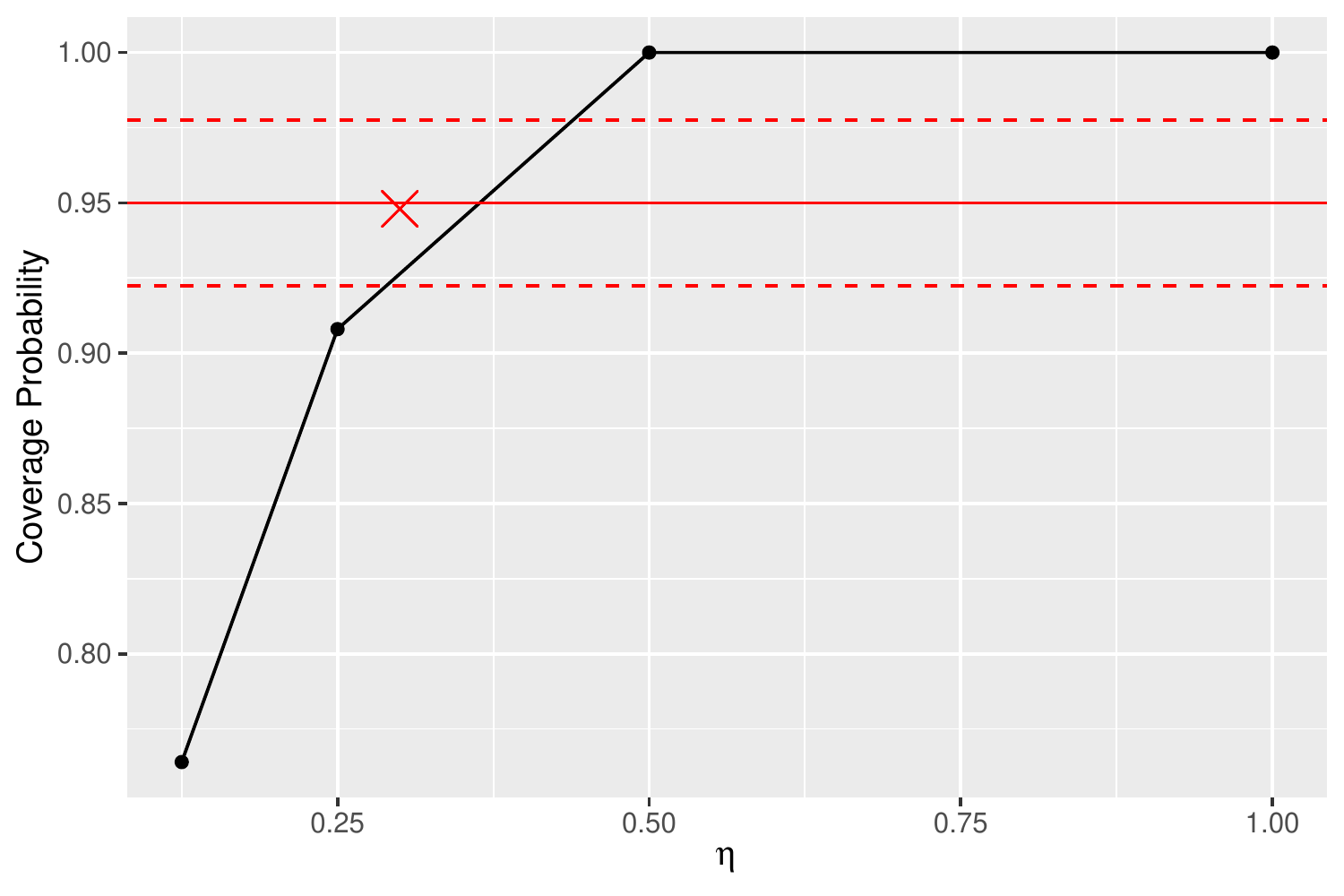}}}
\subfigure[Example 3]{\scalebox{0.5}{\includegraphics{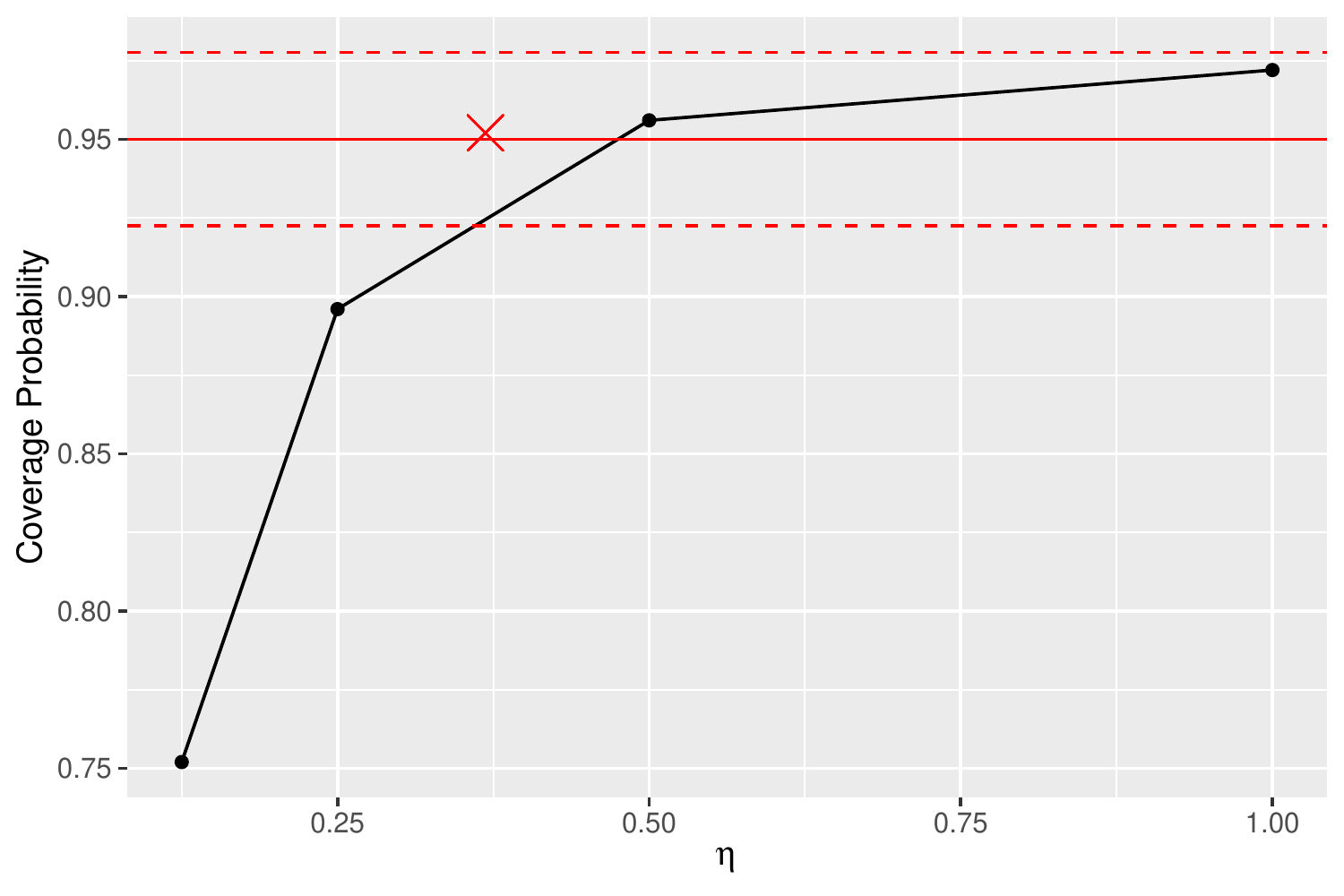}}}
\subfigure[Example 4]{\scalebox{0.5}{\includegraphics{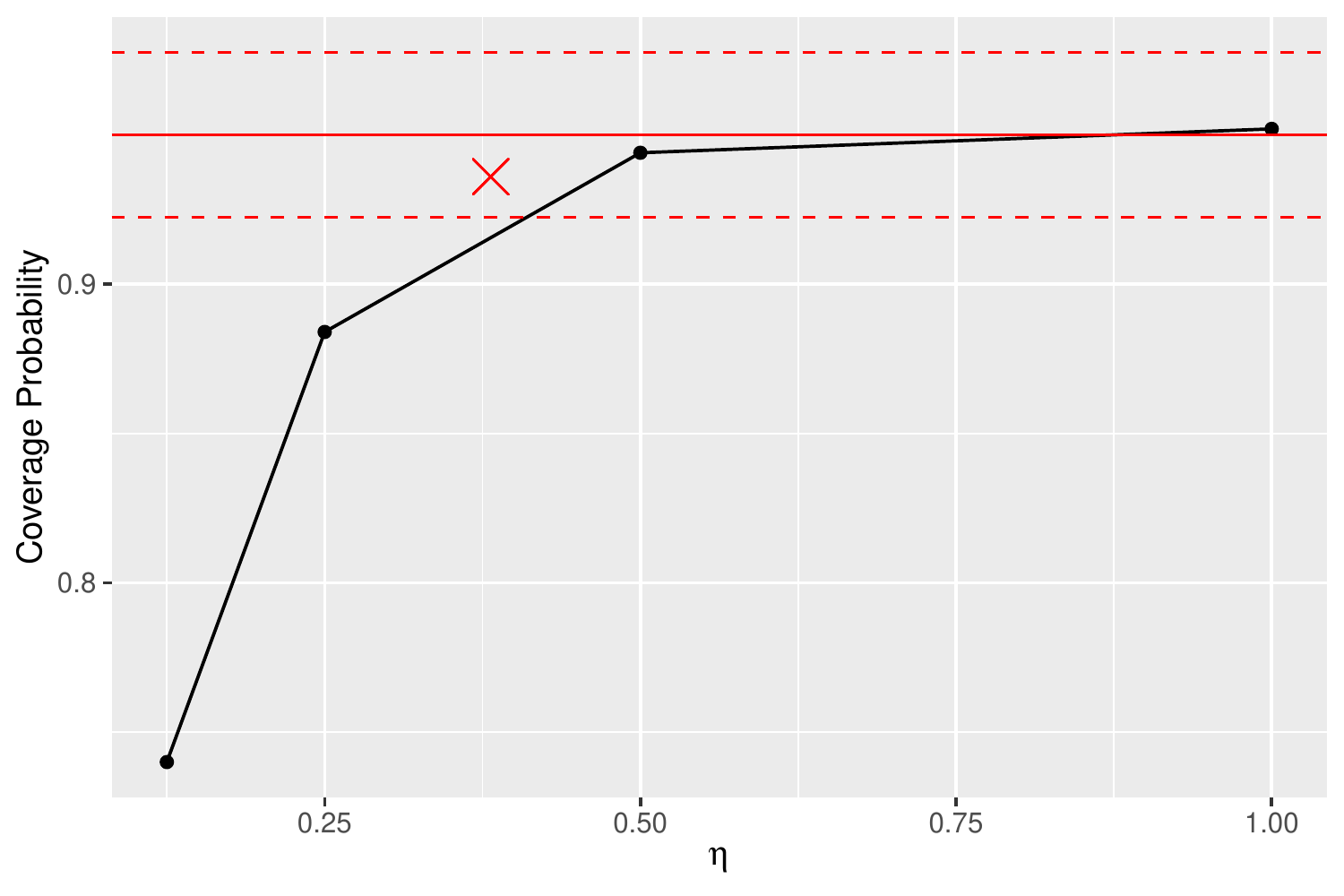}}}
\end{center}
\caption{Coverage probability of the generalized posterior credible interval for the MCID in the four examples as a function of $\eta$. Solid horizontal lines in both plots represent the target 95\% coverage probability, and the dashed line represent the two Monte Carlo standard errors around the target 0.95. The red X corresponds to the average $\hat\eta$ value and the empirical coverage probability attained by our  BQR + GPC credible intervals.}
\label{fig:MCID_lr}
\end{figure}


As we see in Table~\ref{table:MCID_perosnal}, the coverage probability of BQR + GPC is, indeed, right on the advertised level. The asymptotic confidence interval proposed in \cite{zhou2021statistical}, on the other hand, is far less reliable.  Only in Example~1 does it achieve the nominal coverage probability---in the other examples, it falls well short of the 0.95 target.  The bias of the two estimator is comparable in the first two examples, but Zhou et al.'s estimator is much more severely biased in the second two examples, which is sure to be a major contributor to the under-coverage of the corresponding confidence intervals.  Overall, it is clear that our proposed BQR + GPC-based generalized posterior inference is superior in these four examples.

\begin{table}[t]
\begin{center}
\begin{tabular}{ccccccc}
 \toprule
Example & Method & $\eta$ & Coverage & Interval length & Bias & MSE \\
\hline
1& \cite{zhou2021statistical} & -- & 0.97& 0.093 (0.101) & 0.0013 & 0.0002\\
& BQR + GPC &  0.02  & 0.94 & 0.041 (0.008) & 0.0007 & 0.0001 \\
\hline
2& \cite{zhou2021statistical} & -- & 0.51  & 0.185 (0.279)& 0.0073 & 0.0479\\
& BQR + GPC & 0.30 & 0.95 & 0.402 (0.055) & 0.0155 & 0.0094\\
\hline
3 & \cite{zhou2021statistical} & -- & 0.41& 0.259 (0.281) & 0.130 & 0.080\\
& BQR + GPC &  0.37  & 0.95 & 0.794 (0.139) & 0.034 & 0.044 \\
\hline
4& \cite{zhou2021statistical} & -- & 0.34 & 0.179 (0.290) & 0.206 & 0.545\\
& BQR + GPC & 0.38 & 0.94 & 0.924 (0.206) & 0.06 & 0.079\\
  \bottomrule 
\end{tabular}
\end{center}
\caption{Comparison of the method in \cite{zhou2021statistical} with our proposed method (BQR + GPC) in terms of coverage probability and mean length of the 95\% interval estimates for the MCID, as well as bias and mean square error (MSE), for Examples~1--4 described in the text, based on sample size $n=200$ and 250 replications.}
\label{table:MCID_perosnal}
\end{table}

\section{Conclusion}
\label{S:discuss}

In this paper we proposed a new strategy for inference on the personalized MCID, use a generalized posterior distribution based on a purposely-misspecified binary quantile regression model.  This formulation is advantageous for at least two reasons.  First, like the various existing strategies in the literature, it appropriately smooths the (empirical version) of the expected loss that defines the MCID.  Second, this smoothing is done in such a way that it admits a convenient data-augmentation-based posterior sampling algorithm.  This model formulation depends on a parameter $\eta$, which, despite needing to be handled carefully, turns out to be a blessing.  Indeed, it is through careful tuning of this $\eta$ scale parameter, through the proposed GPC strategy, that we can ensure the generalized posterior credible intervals are reliable in a frequentist sense.  As our simulation experiments make clear, the personalized MCID problem is challenging, so achieving even approximately valid uncertainty quantification is apparently no small feat.   



There remains some interesting open questions that deserve further investigation.  First, while the proposed method has strong empirical performance, there is still no formal proof that the GPC method provides valid uncertainty quantification.  

Second, asymptotic properties of our generalized posterior would also be interesting to explore.  The asymptotics for the original binary quantile regression method proposed by Manski are non-standard and highly non-trivial \citep{manski1985semiparametric,kim1990cube,benoit2012binary,mollica2017bayesian} and, while smoothing the objective function can tame some of that, the effect of even some basic smoothing like in \citep{zhou2020interval} is not yet fully understood in this application.  It would be interesting to see if our strategy of smoothing via a purposely-misspecified model would have any general advantages when it comes to asymptotic theory. 

Finally, it would be of interest to see if/how our proposed method for solving the MCID problem could be extended to higher dimensions. Bayesian lasso binary quantile regression in \cite{benoit2013bayesian}, together with GPC, is one reasonable idea. That is, can a suitable choice of $\eta$ ensure the binary quantile regression with a suitable multivariate Laplace prior on $\beta$ ensure that credible interval of MCID achieves the nominal coverage probability? In the higher dimensional MCID problem, our proposed strategy would surely require paralleling the bootstrap replicates in the GPC algorithm, along with perhaps other efficiency-focused adjustments.

\section*{Acknowledgments}

This work is partially supported by the U.S.~National Science Foundation, grants DMS--1811802 and SES--2051225.

\appendix

\section{Appendix}

\subsection{Proof of Proposition~\ref{prop:consistent}}
\label{S:proof}

Recall that the goal is to characterize the minimizer of the BQR-based risk function $\beta \mapsto R_\eta^\text{\sc bqr}(\beta)$, which is defined as 
\[ R_\eta^\text{\sc bqr}(\beta) = \E\{ \ell_\beta^\text{\sc bqr}(X,Y,Z) \}, \]
where 
\begin{align*}
\ell_\beta^\text{\sc bqr}(x,y,z) & = -\log h_\beta(x,y,z) \\
& = \tfrac12\bigl[ (1 + y) \log g_\beta(x,z) + (1 - y) \log\{1 - g_\beta(x,z)\} \bigr] + \log p(x,z), 
\end{align*}
with the form of $g_\beta$---including the implicit dependence on $\eta$ and on $\tau$---given explicitly in \eqref{eq:bqr}.  Since our goal is to minimize this risk function with respect to $\beta$, and since the joint density $p(x,z)$ of $(X,Z)$ under $\prob$ does not depend on $\beta$, we henceforth ignore this term in the loss function.  After dropping this additive constant, our goal is to minimize 
\[  R_\eta^\text{\sc bqr}(\beta) = -\{ A(\beta) + B(\beta) + C(\beta) + D(\beta)\}, \]
where 
\begin{align*}
A(\beta) & = \int_\ZZ \int_{-\infty}^{\beta^\top z} \{\log(1-\tau) - \tau (\beta^\top z-x)/\eta\} \, m_z(x) \, p(x,z) \, dx \, dz \\
B(\beta) & = \int_\ZZ \int_{\beta^\top z}^\infty \log\{1 - \tau e^{-(1-\tau) (x-\beta^\top z)/\eta}\} \, m_z(x) \, p(x,z) \, dx \, dz \\
C(\beta) & = \int_\ZZ \int_{-\infty}^{\beta^\top z} \log\{1 - (1-\tau) e^{-\tau(\beta^\top z-x)/\eta}\} \, \{1-m_z(x)\} \, p(x,z) \, dx \, dz\\
D(\beta) & = \int_\ZZ \int_{\beta^\top z}^\infty \{\log\tau - (1-\tau)(x-\beta^\top z)/\eta\} \, \{1-m_z(x)\} \, p(x,z) \, dx \, dz. 
\end{align*}
It is straightforward, albeit tedious, to differentiate each of these terms with respect to $\beta$ (using Leibniz's integral rule), and we obtain
\begin{align*}
\dot A(\beta) & = \log(1-\tau) \int_\ZZ z \, m_z(\beta^\top z) \, p(\beta^\top z, z) \, dz \\
& \qquad -\frac{\tau}{\eta} \int_\ZZ z \, \int_{-\infty}^{\beta^\top z} m_z(x) \, p(x,z) \, dx \, dz \\ 
\dot B(\beta) & = -\log(1 - \tau) \int_\ZZ z \, m_z(\beta^\top z) \, p(\beta^\top z, z) \, dz \\
& \qquad - \frac{1-\tau}{\eta} \int_\ZZ z \, \int_{\beta^\top z}^\infty \frac{\tau e^{-(1-\tau)(x - \beta^\top z)/\eta}}{1 - \tau e^{-(1-\tau)(x - \beta^\top z)/\eta}} \, m_z(x) \, p(x,z) \, dx \, dz, 
\end{align*}
and, similarly, 
\begin{align*}
\dot C(\beta) & = \log\tau \int_\ZZ z \, \{1-m_z(\beta^\top z)\} \, p(\beta^\top z, z) \, dz \\
& \qquad + \frac{\tau}{\eta} \int_\ZZ z \, \int_{-\infty}^{\beta^\top z} \frac{(1-\tau) e^{-\tau(\beta^\top z - x)/\eta}}{1 - (1-\tau) e^{-\tau(x - \beta^\top z)/\eta}} \, \{1-m_z(x)\} \, p(x,z) \, dx \, dz \\
\dot D(\beta) & = -\log\tau \int_\ZZ z \, \{1 - m_z(\beta^\top z)\} \, p(\beta^\top z, z) \, dz \\
& \qquad + \frac{1-\tau}{\eta} \int_\ZZ z \, \int_{\beta^\top z}^\infty \{1 - m_z(x)\} \, p(x,z) \, dx \, dz. 
\end{align*}
Note that when we sum these four derivatives, the leading terms---those that only involve integration over $\ZZ$---cancel out, providing a little bit of simplification.  From here, it is not difficult to see that 
\begin{equation}
\label{eq:eq}
\dot R_\eta^\text{\sc bqr}(\beta) = 0 \iff \text{LHS} = \text{RHS}(\eta),
\end{equation}
where 
\begin{align*}
\text{LHS} & = (1-\tau) (1-\varpi) \int_\ZZ z \, \int_{\beta^\top z}^\infty \psi_{-1,z}(x) \, dx \, dz - \tau \varpi \int_\ZZ z \, \int_{-\infty}^{\beta^\top z} \psi_{+1,z}(x) \, dx \, dz \\
\text{RHS}(\eta) & = \tau(1-\varpi) \int_\ZZ z \, \int_{-\infty}^{\beta^\top z} \frac{(1-\tau) e^{-\tau(\beta^\top z - x)/\eta}}{1 - (1-\tau) e^{-\tau(x - \beta^\top z)/\eta}} \, \psi_{-1,z}(x) \, dx \, dz \\
& \qquad - (1-\tau) \varpi \int_\ZZ z \, \int_{\beta^\top z}^\infty \frac{\tau e^{-(1-\tau)(x - \beta^\top z)/\eta}}{1 - \tau e^{-(1-\tau)(x - \beta^\top z)/\eta}} \, \psi_{+1,z}(x) \, dx \, dz,
\end{align*}
and where Bayes's formula has been used to introduce a joint distribution of $(X,Y)$, given $Z=z$, of the form $\varpi \psi_{+1,z}(x)$ and $(1-\varpi) \psi_{-1,z}(x)$.  

It is not at all obvious from looking at the above expressions, but it turns out that the risk function is convex.  Specifically, this implies that if we can find a $\beta$ that solves the equation $\text{LHS} = \text{RHS}(\eta)$, then that is the only such solution.  That the risk function is convex might be expected based on the plots shown in Figure~\ref{fig:mcid.risk.comp}.  A detailed proof requires some lengthy expressions, but the idea is clear so we only give a sketch.  Start with 
\[ -\dot A(\beta) = \frac{\varpi \tau}{\eta} \int_\ZZ z \, \int_{-\infty}^{\beta^\top z} \psi_{+1,z}(x) \, dx \, dz + (\text{cancels}). \]
The parenthetic term at the end can be ignored because it gets canceled when the four derivatives are summed up.  The derivative of the above expression with respect to $\beta$ is 
\[ \frac{\varpi \tau}{\eta} \int_\ZZ z z^\top \, \psi_{+1,z}(\beta^\top z) \, dz. \]
Next, look at 
\[ -\dot B(\beta) = \frac{\varpi(1-\tau)}{\eta} \int_\ZZ z \, \int_{\beta^\top z}^\infty \frac{\tau e^{-(1-\tau)(x - \beta^\top z)/\eta}}{1 - \tau e^{-(1-\tau)(x - \beta^\top z)/\eta}} \, \psi_{+1,z}(x) \, dx \, dz + (\text{cancels}). \]
The derivative of this expression with respect to $\beta$ is 
\[ \frac{\varpi(1-\tau)}{\eta} \int_\ZZ z z^\top \, \psi_{+1,z}(\beta^\top z) \, dz + \frac{\varpi(1-\tau)}{\eta} \int_\ZZ z \, \int_{\beta^\top z}^\infty \bigl\{ \cdots \} \, \psi_{+1,z}(x) \, dx \, dz, \]
where the term in curly braces is the derivative of the ratio in the previous display with respect to $\beta$.  On the stated range of integration, that ratio is an increasing function of $\beta$, so the derivative is a positive function times $z^\top$.  Putting these two pieces together
\begin{align*}
-\ddot A(\beta) - \ddot B(\beta) & = \frac{\omega}{\eta} \int_\ZZ z z^\top \, \psi_{+1,z}(\beta^\top z) \, dz \\
& \qquad + \frac{\omega(1-\tau)}{\eta} \int_\ZZ zz^\top \int_{\beta^\top z}^\infty \{\text{positive function}\} \, \psi_{+1,z}(x) \, dx \, dz, 
\end{align*}
which is positive definite.  The same thing happens when we consider $\dot C(\beta)$ and $\dot D(\beta)$, so we conclude that the second derivative of the risk is positive definite, hence the convexity---and uniqueness of the solution---claim follows.  

Back to solving the equation in \eqref{eq:eq}, as the notation above indicates, only the right-hand side term depends on $\eta$, and the claim is that $\text{RHS}(\eta) \to 0$ as $\eta \to 0$.  Showing this will prove the first claim of Proposition~\ref{prop:consistent}.  The key point is that those two ratios involving exponential terms are bounded functions of $\eta$ and both converge pointwise to 0 as $\eta \to 0$.  Then it follows from the dominated convergence theorem that $\text{RHS}(\eta) \to 0$ as $\eta \to 0$.  Therefore, the risk minimizer is Fisher consistent for the (unique) solution to the equation \eqref{eq:mess}, which proves the first claim of the proposition.  

For the proposition's second claim, when $\tau=1-\varpi$, the expression in \eqref{eq:mess} simplifies to 
\[ \int_\ZZ z \int_{-\infty}^{\beta^\top z} \psi_{+1, z}(x) \, dx \, dz = \int_\ZZ z \int_{\beta^\top z}^\infty \psi_{-1,z}(x) \, dx \, dz. \]
If Assumption~S holds, then it is easy to see that 
\[ \int_{-\infty}^{\beta^{\star \top}z} \psi_{+1,z}(x) \, dx = \int_{\beta^{\star\top} z}^\infty \psi_{-1,z}(x) \, dx, \quad \text{for all $z$}, \]
where $\beta^\star = \frac12(a_{+1} + a_{-1})$ is the MCID-determining coefficient presented in \citet{zhou2021statistical}.  This establishes the claimed asymptotic Fisher consistency of our proposed risk minimizer at Zhao et al.'s MCID.  

Finally, when $\tau=1-\varpi=\frac12$, we have that $\text{RHS}(\eta)$ can be simplified as 
\[ \frac{1}{4\eta} \int_\ZZ z \int_{-\infty}^\infty \frac{\frac12 e^{-|x - \beta^\top z|/2\eta}}{1-\frac12 e^{-|x - \beta^\top z|/2\eta}} \, \bigl\{\psi_{-1,z}(x) \, 1_{x \leq \beta^\top z} - \psi_{+1,z}(x) \, 1_{x \geq \beta^\top z} \bigr\} \, dx \, dz. \]
By the symmetry in Assumption~S, along with the symmetry of that ratio inside the integral, it follows that the inside integral vanishes when $\beta=\beta^\star$ for all $z$ and for all $\eta$.  Therefore, in this fully symmetric case, $\beta^\star$ is the minimizer of $R_\eta^\text{\sc bqr}(\beta)$ for all $\eta > 0$.

\subsection{Additional simulation results}
\label{S:more}

In the so-called balanced case, where $\varpi=\frac12$, Proposition~\ref{prop:consistent} states that the minimizer of the risk function $R_\eta^\text{\sc bqr}$ is the true MCID for all $\eta > 0$.  In the unbalanced case, i.e., $\varpi \neq \frac12$, we can conclude that the risk minimizer is the true MCID only for {\em sufficiently small $\eta > 0$}, and for $\eta$ that are not ``sufficiently small,'' the shape of that risk function would be different.  That the risk function's shape depends on $\eta$ implies that, when we use the empirical version of that risk to construct our generalized posterior distribution, both the center and the spread may depend on $\eta$ in a potentially non-trivial way.  And if the center of the generalized posterior varies with $\eta$, then a slight modification to the GPC algorithm, as first presented in \cite{syring2019calibrating}, is needed.  We discussed that slight modification in Section~\ref{SS:pmcid.gpc} and its justification was based on the claim that the generalized posterior center's dependence on $\eta$ is rather mild, so our proposed modification is no more than a superficial change and should not impact the strong performance of GPC as demonstrated elsewhere.  Here we provide some numerical evidence to justify that aforementioned claim.  

Here we revisit the population MCID example in Section~\ref{SS:fisher}, and carry out a simulation study 
to demonstrate that the center of the generalized posterior distribution, as measured by the corresponding posterior mean, depends very mildly on the scalar $\eta$.  Our choice to focus on the population MCID case is to make it easy to draw plots to justify our claims.  Consider an imbalanced case where the joint distribution of $(X,Y)$ is described hierarchically as
\begin{align*}
Y & \sim 2 \, {\sf Ber}(0.7) - 1, \\
(X \mid Y=y) & \sim \nm(y, 2^2), \qquad y = \pm 1.
\end{align*}
As explained in Section~\ref{SS:fisher}, the true (population) MCID being estimated is $\theta^\star=0$.  Throughout the following investigations, each sample consists of $n=200$ pairs of $(X,Y)$ generated from the above joint distribution, and this process is replicated 250 times.  

First, we want to investigate the sensitivity of the generalized posterior distribution's center to the choice of $\eta$.  Figure~\ref{f:populationMCID_posttheta} shows histograms of the 250 generalized posterior means for three different $\eta$ values, namely, $\eta=1$, 0.25, and 0.06, from left to right.  When $\eta$ is relatively large, it is clear that there is some degree of bias in the posterior mean; however, as Proposition~\ref{prop:consistent} predicts, when $\eta$ is small, the generalized posterior mean zeros in on the true MCID.  So, as long as $\eta$ is relatively small, the bias induces by the dependence on $\eta$ is rather mild, as we claimed.  Second, it is of interest to consider what $\eta$ value the GPC algorithm is targeting.  Figure~\ref{f:populationMCID_coverage_hist}(a) plots the coverage probability of the 95\% generalized posterior credible interval for $\theta^\star$ as a function of $\eta$.  Clearly, the coverage probability is of the nominal level when $\eta$ is in the vicinity of 0.5 or 0.6.  Figure~\ref{f:populationMCID_coverage_hist}(b) plots a histogram of the $\hat\eta$ values that were actually selected by (our slight variation on) the GPC algorithm over the 250 replications and, indeed, these concentrate in the 0.4--0.6 range, around 0.47 on average.  Therefore, we conclude that the coverage probability realized by the proposed BQR + GPC generalized posterior credible intervals should be approximately equal to the nominal level, as confirmed in the simulation results presented in Section~\ref{SS:personalizedMCID_results} above. 



\begin{figure}[t]
    \centering
    \includegraphics[scale=0.4]{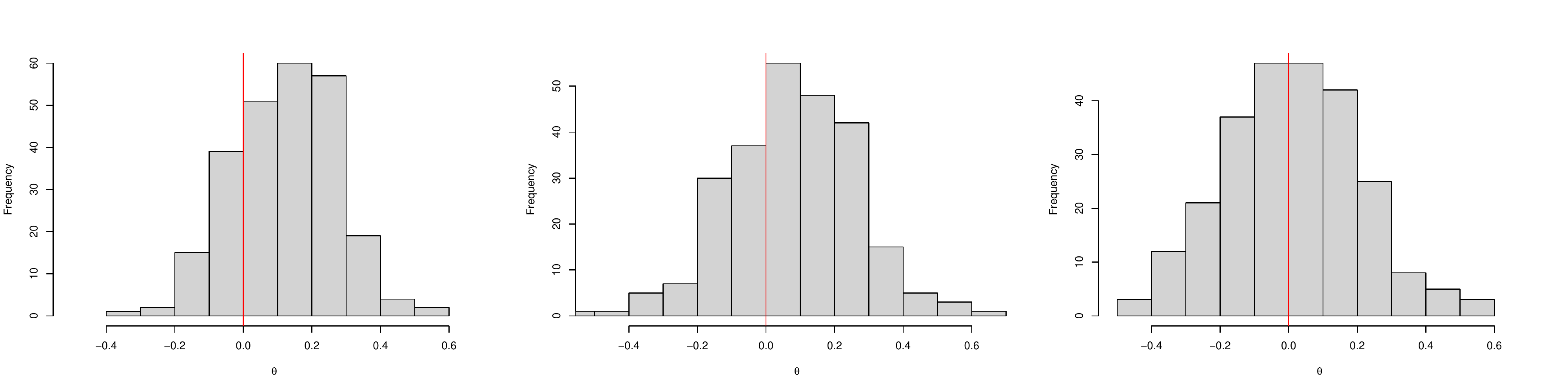}
    \caption{Histogram of the posterior mean over 250 replicates. From left to right $\eta=(1.0,0.25,0.06)$, and the red lines represent the true population MCID.}
    \label{f:populationMCID_posttheta}
\end{figure}

\begin{figure}[t]
\begin{center}
\subfigure[]{\scalebox{0.5}{\includegraphics{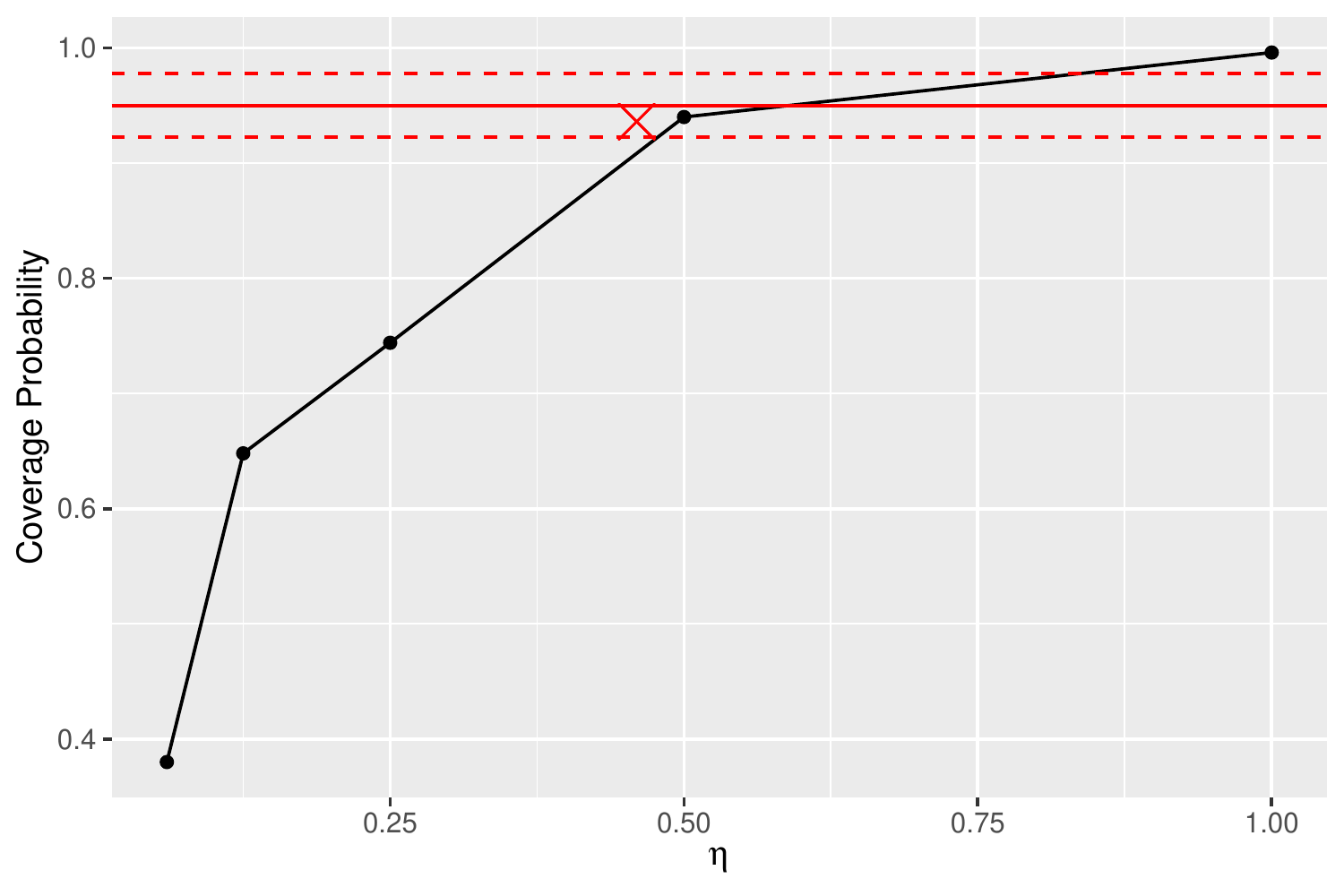}}}
\subfigure[]{\scalebox{0.5}{\includegraphics{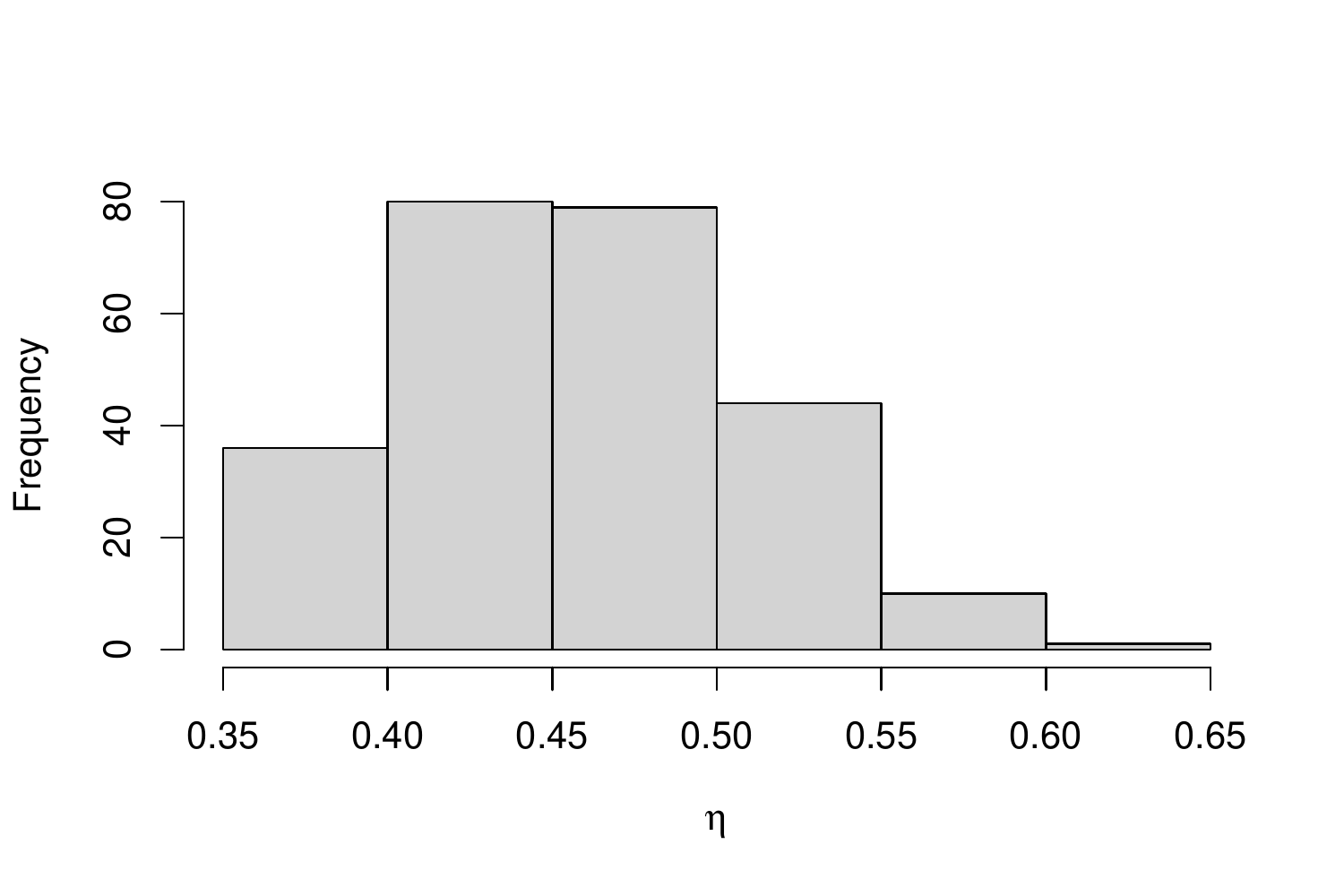}}}
\end{center}
\caption{Panel (a): Coverage probability of the 95\% generalized posterior credible interval for $\theta^\star$ as a function of $\eta$, with tolerable Monte Carlo error bounds around the target 0.95; the red X corresponds to the average $\hat\eta$ value and the empirical coverage probability attained by our BQR + GPC credible intervals. Panel (b): Histogram of the GPC-selected $\hat\eta$ values over 250 replicates.}
\label{f:populationMCID_coverage_hist}
\end{figure}

\bibliography{sample}
\end{document}